\documentclass[a4paper,11pt,reqno]{amsart}
\usepackage[left=25mm,top=25mm,bottom=25mm,right=25mm]{geometry}
\usepackage{dsfont,amscd}
\usepackage[cyr]{aeguill}
\usepackage[pdfdisplaydoctitle=true,
            colorlinks=true,
            urlcolor=blue,
            citecolor=blue,
            linkcolor=blue,
            pdfstartview=FitH,
            pdfpagemode=None,
            bookmarksnumbered=true]{hyperref}
\usepackage{graphicx}
\usepackage{todonotes}

\newtheorem{thm}{Theorem}[section]
\newtheorem{cor}[thm]{Corollary}
\newtheorem{lem}[thm]{Lemma}
\newtheorem{prop}[thm]{Proposition}
\newtheorem{hyp}[thm]{Hypothesis}

\theoremstyle{definition}

\theoremstyle{remark}
\newtheorem{rem}[thm]{Remark}

\newcommand{\CC}{\mathbb{C}}                % Complex numbers
\newcommand{\RR}{\mathbb{R}}                % Real numbers
\newcommand{\ZZ}{\mathbb{Z}}                % Relative integers
\newcommand{\Ela}{\mathbb{E}\mathrm{la}}    % Space of elasticity tensors
\newcommand{\HH}{\mathbb{H}}                % Space of harmonic tensors

\newcommand{\GL}{\mathrm{GL}}               % General linear group
\newcommand{\SO}{\mathrm{SO}}               % Special orthogonal group
\newcommand{\OO}{\mathrm{O}}                % Orthogonal group
\newcommand{\GS}{\mathfrak{S}}              % Symmetric group

\newcommand{\DD}{\mathbb{D}}                % Dihedral group
\newcommand{\octa}{\mathbb{O}}              % Cubic (octahedral) group
\newcommand{\ico}{\mathbb{I}}               % Icosahedral group
\newcommand{\tetra}{\mathbb{T}}             % tetrahedral group
\newcommand{\triv}{\mathds{1}}              % trivial group

\newcommand{\gl}{\mathfrak{gl}}             % General linear algebra
\newcommand{\card}{\operatorname{card}}
\newcommand{\diag}{\operatorname{diag}}

\newcommand{\Id}{\mathrm{Id}}

\newcommand{\NF}{\operatorname{NF}}

\newcommand{\tr}{\operatorname{tr}}
\newcommand{\dev}{\operatorname{dev}}

\newcommand{\norm}[1]{\lVert#1\rVert}                   % norm
\newcommand{\abs}[1]{\lvert#1\rvert}                    % modulus
\newcommand{\set}[1]{\left\{#1\right\}}                 % set
\newcommand{\sprod}[2]{\langle #1 , #2 \rangle}         % scalar product
\newcommand{\strata}[1]{\Sigma_{[#1]}}	                % open stratum
\newcommand{\cstrata}[1]{\overline{\Sigma}_{[#1]}}	    % closed stratum

\newenvironment{menumerate}{%
    	\begin{enumerate}} {\end{enumerate}}

\hypersetup{
    pdftitle={Invariant-based approach to symmetry detection of elasticity tensors
}, %  Title
    pdfauthor={N. Auffray, B. Kolev and M. Petitot}, % Author
    pdfsubject={MSC: 74B05, 15A72}, % Subject
    pdfkeywords={Anisotropy, Symmetry classes, Invariant theory}, % Keywords
    pdflang=EN % Language
    }

\begin{document}

\title{Invariant-based approach to symmetry class detection}

%  Information for first author
\author{N. Auffray}
\address{LMSME, Universit\'{e} Paris-Est, Laboratoire Mod\'{e}lisation et Simulation Multi Echelle, MSME UMR 8208 CNRS, 5 bd Descartes, 77454 Marne-la-Vall\'{e}e, France}
\email{Nicolas.auffray@univ-mlv.fr}

%  Information for second author
\author{B. Kolev}
\address{LATP, CNRS \& Universit\'{e} de Provence, 39 Rue F. Joliot-Curie, 13453 Marseille Cedex 13, France}
\email{kolev@cmi.univ-mrs.fr}

%  Information for third author
\author{M. Petitot}
\address{LIFL, Universit\'{e} des Sciences et Technologies de Lille I, 59655 Villeneuve d'Ascq CEDEX, France}
\email{Michel.Petitot@lifl.fr}

%\thanks{} % Acknowledgements
\subjclass[2000]{74B05, 15A72}%
\keywords{Anisotropy, Symmetry classes, Invariant theory}%

\date{\today}%
%\dedicatory{}%
%\commby{}%

% ----------------------------------------------------------------

\begin{abstract}
In this paper, the problem of the identification of the symmetry class of a given tensor is asked. Contrary to classical approaches which are based on the spectral properties of the linear operator describing the elasticity, our setting is based on the invariants of the irreducible tensors appearing in the harmonic decomposition of the elasticity tensor \cite{FV96}. To that aim we first introduce a  geometrical description of $\Ela$, the  space of elasticity tensors. This framework is used to derive invariant-based conditions that characterize symmetry classes. For low order symmetry classes, such conditions are given on a triplet of quadratic forms extracted from the harmonic decomposition of $C$, meanwhile for higher-order classes conditions are provided in terms of elements of $\HH^{4}$, the higher irreducible space in the decomposition of $\Ela$. Proceeding in such a way some well known conditions appearing in the Mehrabadi-Cowin theorem for the existence of a symmetry plane \cite{CM87} are retrieved, and  a set of algebraic relations on $\HH^{4}$ characterizing the orthotropic ($[\DD_{2}]$), trigonal ($[\DD_{3}]$), tetragonal ($[\DD_{4}]$), transverse isotropic ($[\SO(2)]$) and cubic ($[\octa$]) symmetry classes are provided. Using a genericity assumption on the elasticity tensor under study, an algorithm to identify the symmetry class of a large set of tensors is finally provided.
\end{abstract}

\maketitle

% --------------------------- Table of content ----------------

\begin{scriptsize}
\setcounter{tocdepth}{2}
\tableofcontents
\end{scriptsize}

% ----------------------------------------------------------------
% ----------------------------------------------------------------

\section*{Introduction}
\label{sec:introduction}

% ----------------------------------------------------------------

\subsection*{Physical motivation}

In the theory of linear elasticity, the \emph{stress tensor} $\sigma$ and the \emph{strain tensor} $\varepsilon$ are related, at a fixed temperature, by the \emph{Hooke's law}
\begin{equation*}\label{eq:hooke}
    \sigma^{ij} = C^{ijkl}\varepsilon_{kl}
\end{equation*}
If the medium under study is homogeneous, its elastic behavior is fully characterized by a 4th-order \emph{elasticity tensor}  $\mathbf{C}$. The infinitesimal \emph{strain tensor} $\varepsilon$ is defined as the symmetric displacement gradient
\begin{equation*}
    \varepsilon_{ij}=\frac{1}{2}(u_{i,j}+u_{j,i})
\end{equation*}
and, assuming that the material is not subjected to volumic couple, the associated stress tensor is the classical symmetric Cauchy's one.
Hence, both $\sigma$ and $\varepsilon$ belong to the $6$-D vector space $S^{2}(\RR^{3})$, where $S^{2}$ denotes the symmetric tensor product. As a consequence, the elasticity tensor is endowed with \emph{minor symmetries}:
\begin{equation*}
    C^{ijkl} = C^{jikl} = C^{ijlk} .
\end{equation*}
In the case of hyperelastic materials, the stress-strain relation is furthermore assumed to derive from an elastic potential. Therefore the elasticity tensor is the second-derivative of the potential energy with respect to strain tensor. Thus, as a consequence of Schwarz's theorem, the hyperelasticity tensor possesses the \emph{major symmetry}:
\begin{equation*}
    C^{ijkl}= C^{klij} .
\end{equation*}
The space of  elasticity tensors is therefore the $21$-D vector space $\Ela = S^{2}S^{2}(\RR^{3})$. To each hyperelastic material corresponds an elasticity tensor  $\mathbf{C}$ but this association is not unique, there is a \emph{gauge group}. Indeed, a designation based on the components of  $\mathbf{C}$ in a fixed reference frame is relative to the choice of a \emph{fixed orientation} of the material. If this problem is invisible for isotropic materials, it becomes more prominent as the anisotropy increases. For a fully anisotropic (triclinic) material how can it be decided whether two sets of components represent the same material ? To address this problem some approaches have been proposed in the literature so far. Roughly speaking, those propositions are based on the computation of the elasticity tensor spectral decomposition. But despite of being theoretically well-founded, the symmetry class identification relies on the evaluation of the roots multiplicity of the characteristic polynomial. Such a problem is known to be highly ill-posed \cite{Zen05}. Therefore this approach is difficult to handle in practice especially working with noise corrupted data.

In order to avoid the computation of the spectral decomposition an alternative approach based on the \emph{invariants} (and \emph{covariants}) of the elasticity tensor is presently considered. Such a way to proceed is in the perspective of Boheler and al. \cite{BKO94} original work. It is interesting to notice that the general framework of the present paper is already well-known by the high energy physics community \cite{AS81,AS83,SV03}. However, these methods do not seem to have been yet applied in elasticity  and especially for a representation as sophisticated as $\HH^{4} (\RR^{3})$, which \emph{invariant algebra} is not free.
Reconstruction method based on this invariant-based approach, and numerical comparisons between approaches will be the objects of forthcoming papers.
Before introducing in depth the proposed method, we wish to draw a short picture of this approach. Therefore, the next paragraph will be devoted to present  the  method in a nutshell.

% ----------------------------------------------------------------

\subsection*{Invariant-based identification in a nutshell}
\label{subsec:normal_forms_nutshell}

As the material is rotated by an element $\mathbf{g}(=g_{i}^{j}) \in \SO(3)$, the elasticity tensor  $\mathbf{C}$ moves under the action of the 3D rotation group $\SO(3)$ on the space of elasticity tensors $\Ela$
\begin{equation*}\label{eq:elastic_representation}
    C^{ijkl} \mapsto g_{p}^{i}\,g_{q}^{j}\,g_{r}^{k}\,g_{s}^{l}\, C^{pqrs} .
\end{equation*}
Hence, from the point of view of linear elasticity, the classification of elastic materials can be assimilated to the description of the \emph{orbits} of $\SO(3)$-action on $\Ela$. Practically, this is a  difficult problem because the orbit space of a Lie group action is not a smooth manifold in general. Usually, the orbit space is a \emph{singular space} which can described by the mean of the \emph{isotropic stratification}~\cite{AS83} (described in \autoref{sec:orbit_spaces}). In the field of elasticity, this \emph{isotropic stratification} was first established by Forte and Vianello \cite{FV96}\footnote{The concept of symmetry class used by Forte and Vianello is equivalent to the notion of an isotropic stratum.}.

Since the rotation group $\SO(3)$ is compact, the algebra of invariant polynomial functions on $\Ela$ is finitely generated~\cite{Wey97} and \emph{separates the orbits}~\cite[Appendix C]{AS83}. Therefore, it is possible to find a finite set of polynomial invariants $\set{J_{1}, \dotsc, J_{N}}$ which define a polynomial function $J: \Ela \to \RR^{N}$ such that
\begin{equation*} \label{eq:J}
    J(\mathbf{C}) = J(\mathbf{C}^{\prime}) \quad \text{iff} \quad \mathbf{C} \approx \mathbf{C}^{\prime}.
\end{equation*}
This way, an $\SO(3)$-orbit can be identified with a point in $J(\Ela) \subset \RR^N$. A classical result from algebraic geometry states that the image of a $\RR$-vector space under a polynomial map is a \emph{semi-algebraic set}, that is a subset of $\RR^{N}$ defined by a boolean combination of polynomial equations and inequalities over $\RR$ \cite{Cos02}. The orbit space is therefore a semi-algebraic set.
The invariants $\set{J_{1}, \dotsc, J_{N}}$ are polynomials in $21$ variables ($\dim \Ela=21$). To proceed towards our goal, we will make an extensive use of the decomposition of $\Ela$ into $\SO(3)$-irreducible components, also known as the \emph{harmonic decomposition} (c.f. \autoref{sec:elasticity_tensors}):
\begin{equation*}
    \Ela := S^{2} S^{2}(\RR^{3}) \simeq 2 \HH^{0} (\RR^{3})  \oplus 2 \HH^{2}(\RR^{3})   \oplus \HH^{4} (\RR^{3})
\end{equation*}
where $\HH^{n}(\RR^{3}) $ is the space of $n$th-order \emph{harmonic tensors} on $\RR^{3}$.  Because of the lack of ambiguity, this notation will be shorten $\HH^{n}$. It is well known that $\dim \HH^n = 2n+1$ for all $n\ge 0$.  Historically this decomposition has already been used in the study of anisotropic elasticity tensors~\cite{Bac70,JCB78,Bae93,FV96}. We therefore refer to these references for a deeper insight into this topic.

It is known \cite{FV96} that $\Ela$ is divided into 8 conjugacy classes for the $\SO(3)$-action. For each isotropy class $[H]$, we build a \emph{linear slice} $\Ela^{H}$ and we show that the equivalence relation $\mathbf{C} \approx \mathbf{C}^{\prime}$ where $\mathbf{C}, \mathbf{C}' \in \strata{H}$ reduces to the linear action of a certain \emph{monodromy group} $\Gamma^{H}$ on $\Ela^{H}$. The cardinal of $\Gamma^{H}$ is equal to the number of intersection points of a generic orbit in $\strata{H}$ with the linear slice $\Ela^{H}$. The closure of each stratum $\strata{H}$, which we denote by $\cstrata{H}$ and refer to as a closed stratum, is a semialgebraic set which we describe explicitly by a finite number of equations and inequalities on the invariant polynomials $\set{J_{1}, \dotsc, J_{N}}$. For each class $[H]$ for which $\Gamma^{H}$ is finite, we give \emph{explicit relations} between invariants to check whether $\mathbf{C}$ belongs to this class.

% ----------------------------------------------------------------

\subsection*{Organization of the paper}

The first section of this paper (\autoref{sec:orbit_spaces}) is devoted to the introduction of the framework uses for this study. There is no claim for any originality here, similar material can be found in \cite{GSS88,SV03}. We have tried to make this paper accessible with
only few prerequisites. The notions of \emph{stratification} and \emph{linear slice} are defined, together with a group theoretical method to compute its associated degree. In the second section (\autoref{sec:stratification_quadratic_forms}), we illustrate these concepts on a simple example: the stratification of the space of quadratic forms on $\RR^{3}$ under the $\SO(3)$-action and we study the stratification of a $n$-uple of quadratic forms. This constitutes an extension of results on the stratification of $2$ quadratic forms. Our results are provided in terms of harmonic tensors, which is not so usual.
In \autoref{sec:elasticity_tensors}, the harmonic decomposition is introduced and the stratification of $\Ela$ discussed. The geometrical concepts introduced in \autoref{sec:orbit_spaces} will be enlighten in view of well-known results of anisotropic elasticity. In \autoref{sec:stratification_H4}, we study the strata of $\HH^{4}$ which have finite monodromy group. In each situation the stratification is established and algebraic criteria identifying symmetry classes are provided. These criteria are, up to authors' best knowledge, new. This section will be closed by a bifurcation diagram that sums up the symmetry breaking from $\HH^4$ together with the algebraic relations that induce those transitions.

According to the restriction to linear slices, only $6$ out of the $8$ elastic symmetry classes possess linear slices with finite monodromy group. Therefore, in the present paper, algebraic relations are only provided for the following anisotropic classes : isotropy ($[\SO(3)]$), cubic ($[\mathbb{O}]$), transverse isotropy ($[\OO(2)]$), tetragonal ($[\mathbb{D}_{4}]$), trigonal ($[\mathbb{D}_{3}]$), orthotropic ($[\mathbb{D}_{2}]$). Nevertheless, the combination of results from \autoref{sec:stratification_quadratic_forms} and \autoref{sec:stratification_H4} completed by an assumption on the genericity of anisotropic of the elasticity tensor understudy (in a sense that will be defined in \autoref{sec:stratifaction_ela}.) allow  the $8$ symmetry classes to be identified. This paper will be concluded in \autoref{sec:conclusion} and the extension of the method to a broader class of situations will be discussed.

% ----------------------------------------------------------------

\subsection*{Main results}

The main results of this paper are
\begin{itemize}
\item In \autoref{sec:elasticity_tensors} a rigorous geometric framework is provided that allows to recover important features of anisotropic elasticity \cite{Nor89}. Furthermore we draw a link between the conditions appearing in the well-known Mehrabadi-Cowin theorem \cite{CM87} and the elasticity tensor harmonic decomposition ;
\item In \autoref{sec:stratification_H4}, five sets of algebraic relations are provided to identify the following symmetry classes in $\HH^{4}$: orthotropic ($[\DD_{2}]$), trigonal ($[\DD_{3}]$), tetragonal ($[\DD_{4}]$), transverse isotropic ($[\SO(2)]$) and cubic ($[\octa$]). Furthermore, we provide a parametrization of the corresponding strata by rational expressions involving up to $6$ polynomial invariants, namely: $J_{2},\ldots,J_{7}$.
\item As an observation, we also prove that the invariants defined by the coefficients of the Betten polynomial \cite{Bet87} do not separate the orbits of $\Ela$.
\end{itemize}

% ----------------------------------------------------------------

\subsection*{Notations}
\label{subsec:notations}

To conclude this introduction let us specify some typographical conventions that will be used throughout this paper. Scalars will be noted with Greek letters ; minuscule letters indicate elements of $\HH^{2}$, i.e deviator ; bold minuscule letters stand for elements of $\RR^{3}\otimes\RR^{3}$, i.e. second order tensors ; elements of $\HH^{4}$ will be indicated by capital letters, meanwhile elasticity tensors (elements of $S^{2}S^{2}(\RR^{3}))$ will be noted  by bold capital letters.
Some minors exceptions to that rule may occur (such as for the stress $\sigma$ and strain $\varepsilon$ tensors) and will therefore be indicated in the text.

% ----------------------------------------------------------------
% ----------------------------------------------------------------

\section{Geometry of orbit spaces}
\label{sec:orbit_spaces}

In this section, we consider a linear representation
\begin{equation*}
    \rho: G \to \GL(V)
\end{equation*}
of a \emph{compact} real Lie group\footnote{In the following $G$ will solely represent a \emph{compact} real Lie group, therefore this precision will mostly be omitted.} $G$ on a finite dimensional $\RR$-vector space $V$. This action will be noted:
\begin{equation*}
    g\cdot v := \rho(g)v
\end{equation*}
Unless a real Lie group $G$ acts without \emph{fixed point} on $V$, the \emph{orbit space} $V/G$ is not a \emph{differentiable manifold} \cite{AS83}. This is the main reason why the description of the orbit space is difficult in general.  However, for a compact Lie group, the algebra of $G$-invariant polynomials, is known to separate the orbit (c.f. \autoref{subsec:separation_orbits}). Using this algebra, $V/G$ can be described as a semialgebraic\footnote{A semialgebraic set is a subset $S$ of $\RR^n$ defined by a finite sequence of polynomial equations (of the form $P(x_{1},\dotsc,x_{n}) = 0$) and inequalities (of the form $Q(x_{1},\dotsc,x_{n})$).} subset of $\RR^N$. Let $\set{J_{1}, J_{2}, \dotsc, J_{N}}$ be a set of polynomial invariants which separate the $G$-orbits. Then, the map
\begin{equation*}
    J : v \mapsto \big(J_{1}(v), J_{2}(v), \dotsc, J_{N}(v)\big)
\end{equation*}
induces an algebraic homeomorphism between the \emph{orbit space} $V/G$ and $J(V)\subset \RR^{N}$ which is a semialgebraic subset of $\RR^N$. In particular, such a set is a \emph{stratified space}. It is the union of manifolds of various dimensions, called \emph{strata}, each of them being the union of orbits of a given \emph{isotropy type}. This construction is the object of the present section.

% ----------------------------------------------------------------

\subsection{G-orbits}
\label{subsec:orbits}

Two vectors $v_{1}$ and $v_{2}$ are said to be $G$-related, and we write $v_{1} \approx v_{2}$ if there exists $g\in G$ such that $v_{2} = g\cdot v_{1}$. The set of all vectors $v\in V$ which are related to $v$ by $G$ is called the $G$-orbit of $v$ and is denoted by $G\cdot v$:
\begin{equation*}
    G\cdot v := \set{g\cdot v \mid g\in G}
\end{equation*}
The $G$-orbits are compact submanifolds of $V$~\cite{AS83,Bre72}.

% ----------------------------------------------------------------

\subsection{Isotropy subgroups}
\label{subsec:isotropy_subgroups}

The set of transformations of $G$ which leave the vector $v$ fixed forms a closed subgroup of $G$, which is called the \emph{isotropy subgroup} (or \emph{symmetry group}) of $v$. It will be denoted by $G_{v}$:
\begin{equation*}
    G_{v} := \set{g\in G;\; g\cdot v = v}
\end{equation*}
It can be shown that symmetry groups of $G$-related vectors are conjugate:
\begin{lem}\label{lem:isotropy_groups}
For any $v\in V$ and any $g\in G$, we have
\begin{equation*}
    G_{g\cdot v} = g\ G_{v}\ g^{-1}
\end{equation*}
\end{lem}

% ----------------------------------------------------------------

\subsection{Isotropy classes}
\label{subsec:isotropy_classes}

Two vectors $v_{1}$ and $v_{2}$ are in the same \emph{isotropy class} and we write $v_{1} \sim v_{2}$ if their isotropy subgroups are \emph{conjugate} in $G$, that is there exists $g\in G$ such that
\begin{equation*}
    G_{v_{2}} = g\,G_{v_{1}}g^{-1}
\end{equation*}
Due to lemma~\ref{lem:isotropy_groups}, two vectors $v_{1}$ and $v_{2}$ which are in the same $G$-orbit are in the same isotropy class but the converse is {generally} false. Therefore, the relation \emph{``to be in the same isotropy class''} is \emph{weaker} than the relation \emph{``to be in the same $G$-orbit''}. This observation \emph{is summed up by} the following implication
\begin{equation*}
    v_{1} \approx v_{2} \Rightarrow v_{1} \sim v_{2}
\end{equation*}
The conjugacy class of a subgroup $H$ will be denoted by $[H]$. The conjugacy class of an isotropy subgroup will be called an \emph{isotropy class}. One can show that there is only a \emph{finite} number of isotropy classes for a representation of a compact Lie group (see~\cite{Bre72}). In the field of physics, for the $\SO(3)$-action on tensor, this result is known as the Hermann theorem \cite{Her45,Auf08a}.
On the set of conjugacy classes of \emph{closed} subgroups of a $G$, there is a partial order induced by inclusion. It is defined as follows:
\begin{equation*}
    [H_{1}] \preceq [H_{2}] \quad \text{if $H_{1}$ is conjugate to a subgroup of $H_{2}$}.
\end{equation*}
\emph{Endowed} with this partial order, the set of isotropy classes is a \emph{finite lattice}\footnote{The lattice of conjugacy classes of all closed subgroups of $\SO(3)$ is described in \autoref{sec:normalizers}.}. In particular, it has a \emph{least element} and a \emph{greatest element}.

% ----------------------------------------------------------------

\subsection{Isotropic stratification}
\label{subsec:isotropic_stratification}

The set of all vectors $v$ in the same isotropy class defined by $[H]$ is denoted $\strata{H}$ and called a \emph{stratum}. To avoid any misunderstanding it worths to note that a stratum is not a vector space, but a fiber bundle.
There is only a \emph{finite} number of (non empty) strata and each of them is a smooth submanifold of $V$~\cite{AS83,Bre72}. The partial order relation on conjugacy classes induces a (reverse) partial order relation on the strata
\begin{equation*}
    [H_{1}] \preceq [H_{2}] \Leftrightarrow \Sigma_{[H_{2}]} \preceq  \Sigma_{[H_{1}]}
\end{equation*}
The set of strata inherits therefore the structure of a finite lattice. An orbit $G\cdot v$ is said to be \emph{generic} if it belongs to the least isotropy class. The \emph{generic stratum} $\strata{H_{0}}$, which is the union of generic orbits can be shown to be a \emph{dense and open} set in $V$. Moreover, for each other stratum $\strata{H}$, we have $\dim \strata{H} < \dim \strata{H_{0}}$.
The partition
\begin{equation*}
    V = \strata{H_{0}} \cup \Sigma_{[H_{1}]} \cup \dotsb \cup \Sigma_{[H_{n}]}
\end{equation*}
is called the \emph{isotropic stratification} of $(V,\rho)$. $\strata{H_{0}}$ is the \emph{generic stratum}. On the opposite side, the stratum which corresponds to the greatest isotropy subgroup is called the \emph{minimum stratum}.

\begin{rem}
  Notice that the closure of a stratum $\strata{H}$, denoted $\cstrata{H}$ and called the \emph{closed stratum} associated to $[H]$, corresponds to vectors $v\in V$ such that $v$ has \emph{at least} isotropy $[H]$, whereas $\strata{H}$ corresponds to vectors $v\in V$ such that $v$ has \emph{exactly} isotropy $[H]$.
\end{rem}

% ----------------------------------------------------------------
% ----------------------------------------------------------------

\subsection{Fixed point sets and normalizers}
\label{subsec:fixed_point_sets}

Let $H$ be any subgroup of $G$. The set of vectors $v\in V$ which are fixed by $H$
\begin{equation*}
    V^{H} := \set{ v\in V \mid h.v = v \text{ for all } h \in H},
\end{equation*}
is called the \emph{fixed point set} of $H$. For each $v\in V^{H}$, $G_{v}\subset H$. One can check that if $H_{1} \subset H_{2}$ then $V^{H_{2}} \subset V^{H_{1}}$. However, it may happen that $V^{H_{2}} = V^{H_{1}}$ but $H_{1} \ne H_{2}$.
Furthermore, it should be pointed out that, in general, given $H_{1}, H_{2}$ in the same conjugacy class $[H]$, we have $V^{H_{2}} \ne V^{H_{1}}$. Notice also that a $G$-orbit belongs to an isotropy class greater than $[H]$ iff it intersects the space $V^{H}$.

Given a subgroup $H$ of $G$, the \emph{normalizer} of $H$, defined by
\begin{equation*}
    N(H) := \set{ g\in G \mid gHg^{-1} = H}
\end{equation*}
is the maximal subgroup of $G$, in which $H$ is a normal subgroup. It can also be described as the \emph{symmetry group} of $H$ when $G$ acts on its subgroups by conjugacy.

\begin{lem}\label{lem:normalizer}
The space $V^{H}$ is $N(H)$-invariant. Moreover, if $H = G_{v_{0}}$ is the \emph{isotropy group} of some point $v_{0}\in V$, then $N(H)$ is the maximal subgroup which leaves invariant $V^{H}$.
\end{lem}

\begin{proof}
Let us first define
\begin{equation*}
    K := \set{g \in G ; \; g \left(V^{H}\right) = V^{H}}
\end{equation*}
Now let $g \in N(H)$ and $v \in V^{H}$. Given $h\in H$ we have $g^{-1}hg\in H$ and thus $g^{-1}hg\cdot v = v$. But then $hg\cdot v = g\cdot v$ and since this is true for all $h\in H$ we get that $g\cdot v \in V^{H}$. That is $V^{H}$ is invariant under $N(H)$, therefore $N(H) \subset K $.

Suppose now that $H = G_{v_{0}}$ for some point $v_{0}\in V$. For all $g \in K$ and $h \in H$, we have $h \cdot g \cdot v_{0} = g \cdot v_{0}$ and hence $g^{-1} \cdot h \cdot g \cdot v_{0} = v_{0}$. Therefore, $g^{-1}hg \in G_{v_{0}} = H$. In other words, $g^{-1}Hg \subset H$ but, since this is also true for $g^{-1}$, we have finally $g^{-1}Hg = H$ and therefore $K \subset N(H)$. Hence $K=N(H)$ which concludes the proof.
\end{proof}

% ----------------------------------------------------------------

\subsection{Linear slices}
\label{subsec:linear_slices}

One consequence of lemma~\ref{lem:normalizer} is that, the linear representation $\rho: G \to \GL(V)$ induces a linear representation, of $N(H)$ on $V^{H}$, obtained by the restriction
\begin{equation*}
    \rho_{N(H)} : N(H) \longrightarrow \GL(V^{H}).
\end{equation*}
This induced representation is not faithful. However, \emph{when $H$ is an isotropy group}, its kernel is exactly $H$ and we get a \emph{faithful} linear representation
\begin{equation*}
    \rho_{\Gamma^{H}}: \Gamma^{H} \longrightarrow \GL(V^{H}) \text{ where } \Gamma^{H} := N(H)/H .
\end{equation*}
Furthermore, in that case, the equivalence relation defined by $(V,\rho)$, when restricted to the set $V^{H}\cap \strata{H}$, is the same as the equivalence relation defined by $(V^{H},\rho_{N(H)})$. Indeed, let $v_{1},v_{2}$ in $V^{H}$, such that $G_{v_{1}} = G_{v_{2}} = H$ and $v_{2} = g \cdot v_{1}$. Then $ghg^{-1}\cdot v_{2} = v_{2}$, for each $h \in H$. Thus $ghg^{-1} \in G_{v_{2}} = H$ and $g \in N(H)$.

We have therefore reduced (locally, on the stratum $\strata{H}$) the problem of describing the orbit space of $V$ mod $G$ by the orbit space of $V^{H}$ mod $\Gamma^{H}$. This is especially meaningful when the group $\Gamma^{H}$ is \emph{finite}, in which case, we say that $V^{H}$ is a \emph{linear slice}. Then, each $G$-orbit intersect $V^{H}$ at most in a finite number of points and the natural map $V^{H} \rightarrow V^{H}/\Gamma^{H}$ is a finite \emph{ramified covering}\footnote{A \emph{ramified covering} is a generalization of the concept of a \emph{covering map}. It was introduced in the theory of Riemann surfaces and corresponds to the situation where a finite (or discrete) group acts on a manifold but with fixed points. A typical example is given by the map $z \mapsto z^{n}$ on the complex plane. Outside $0$, this map is a \emph{covering} but $0$ is a singularity, a \emph{point of ramification}. In higher dimension, the set of ramifications points may be more complicated.} of $V^{H}/\Gamma^{H}$. The \emph{degree} of  $V^{H}$ is defined as the cardinal of $\Gamma^{H}$. It is equal to the index $[N(H),H]$ of $H$ in its normalizer $N(H)$.
It is also the number of points in a generic fiber of the covering $V^{H} \longrightarrow V^{H}/\Gamma^{H}$.
The group $\Gamma^{H}$ is called the \emph{monodromy group} of  $V^{H}$.

% ----------------------------------------------------------------

\subsection{Strata dimensions}
\label{subsec:strata_dimension}

The restriction of the projection map $\pi : V \to V/G$ to a stratum $\strata{H}$
\begin{equation}\label{eq:bundle}
    \pi_{/\strata{H}} : \strata{H} \to \strata{H}/G
\end{equation}
is a \emph{fiber bundle} with fiber $G/H$, i.e. a space that locally looks like $\strata{H}/G\times G/H$. Therefore, we have
\begin{equation*}
    \dim \strata{H} = \dim (G/H) + \dim(\strata{H}/G)
\end{equation*}
where  $\dim (G/H) = \dim G - \dim H$.

Notice that the set $V^{H} \cap \strata{H}$ is the subspace of vectors in $V^{H}$ which isotropy class is exactly $[H]$. It is an open and dense set in $V^{H}$, i.e almost all elements of $V^{H}$ are in $V^{H} \cap \strata{H}$. These vectors correspond to the generic elements in $V^{H}$.

Now, since the equivalence relation defined by $(V,\rho)$, when restricted to the space $V^{H}$, is the same as the equivalence relation defined by $(V^{H},\rho_{\Gamma^{H}})$, the basis $\strata{H}/G$ of the bundle~\eqref{eq:bundle} is diffeomorphic to $(V^{H} \cap \strata{H})/\Gamma^{H}$  (see~\cite{Bre72} for a more rigorous justification). Thus
\begin{equation}\label{eq:quotient_strata_dimenion}
\begin{split}
  \dim\strata{H}/G & = \dim (V^{H} \cap \strata{H})/\Gamma^{H}= \dim V^{H} - \dim \Gamma^{H} \\
  \dim \strata{H} &= \dim V^{H} + \dim (G/ N(H))=\dim V^{H} + \dim G - \dim N(H)
\end{split}
\end{equation}
%Notice that this formula does not require $V^{H}$ to be of finite degree (i.e $\dim \Gamma^{H} = 0$).
Mechanical illustrations of the physical contents of those geometric concepts will be given in \autoref{subsec:isotropic_stratification_ela}.

\begin{rem}
The dimension of $V^{H}$ can be computed using the \emph{trace formula} for the \emph{Reynolds operator} (see~\cite{GSS88}). Letting $\chi_{\rho}$ be the character of the representation $(V,\rho)$, we have
\begin{equation}\label{eq:trace_formula}
    \dim V^{H} = \frac{1}{\abs{H}}\sum_{h \in H} \chi_{\rho}(h)
\end{equation}
if $H$ is a finite group (and the preceding formula has to be replaced by the Haar integral over $H$ for an infinite compact group). Explicit analytical formulas based on \eqref{eq:trace_formula} were obtained in \cite{Auf10} and \cite{GSS88}. They can be found in \autoref{sec:TrFrm}.
\end{rem}

% ----------------------------------------------------------------

\subsection{Implicit equations for closed strata}
\label{subsec:implicitization}

Let $[H]$ be a fixed isotropy class and $\strata{H} \subset V$  the corresponding stratum. Given a finite set of invariant polynomials which separate the orbits of $G$  , the closed stratum $\cstrata{H}$ can be characterized by a finite set of relations (equations and inequalities) on $\set{J_{1}, J_{2}, \dotsc, J_{N}}$. To obtain these equations, we fix a subgroup $H$ in the conjugacy class $[H]$ and choose linear coordinates $(x^{i})_{1 \le i \le q}$ on $V^{H}$. Then, we evaluate $\set{J_{1}, J_{2}, \dotsc, J_{N}}$ on $V^{H}$ (as polynomials in the $x^{i}$) and we try to obtain \emph{implicit equations} on the $J_{k}$ for the parametric system
\begin{equation}\label{eq:parametric_system1}
    J_{k} = P_{k}(x^{1}, \dotsc x^{q}), \quad k = 1, \dotsc , N,\quad P_{k}\in\RR[X_{1},\dotsc,X_{q}]
\end{equation}
satisfied by the restriction of $\set{J_{1}, J_{2}, \dotsc, J_{N}}$ to $V^{H}$.

This is a difficult task in general (one might consider~\cite{CLO07} for a full discussion about the implicitization problem). Moreover, it could happen that the algebraic variety $V_{I}$ defined by such an implicit system is bigger than the variety $V_{P}$, defined by the parametric system~\eqref{eq:parametric_system1}. Besides, we are confronted to the difficulty that we work over $\RR$ which is not algebraically closed.
To overcome these difficulties, we observe that the restrictions of $\set{J_{1}, J_{2}, \dotsc, J_{N}}$ to $V^{H}$ are $\Gamma^{H}$-invariant and can therefore be expressed as polynomial expressions of some generators $\sigma_{1}, \dotsc , \sigma_{r}$ of the invariant algebra of the monodromy group $\Gamma^{H}$ on $V^{H}$. This observation reduces the problem to consider first the \emph{implicitization problem} for the parametric system
\begin{equation}\label{eq:parametric_system2}
    J_{k} = p_{k}(\sigma_{1}, \dotsc , \sigma_{r}), \quad k = 1, \dotsc , N, \quad p_{k}\in\RR[X_{1},\dotsc,X_{r}]
\end{equation}
To solve this problem, we use a \emph{Groebner basis} \cite{CLO07} (see~\autoref{subsec:Grobner}) or a \emph{regular chain}\cite{BLOP95}.
The main observation is that, in each case we consider, the system~\eqref{eq:parametric_system2} leads to a system of relations (syzygies) which characterizes \emph{the closed stratum} $\cstrata{H}$
\begin{equation}\label{eq:method_syzygies}
    S_{j}(J_{1}, J_{2}, \dotsc, J_{N}) = 0, \quad j = 1, \dotsc l\ \quad S_{j}\in\RR[X_{1},\dotsc,X_{N}] ,
\end{equation}
and a system
\begin{equation}\label{eq:method_solutions}
    \sigma_{i} = R_{i}(J_{1}, J_{2}, \dotsc, J_{N}), \quad i = 1, \dotsc , r\ \quad R_{i}\in\RR(X_{1},\dotsc,X_{N}) ,
\end{equation}
which express the $\sigma_{i}$ as rational\footnote{The fact that the solutions are rational will not be justified here. We just observe that this is the case for all classes we have treated so far in this article.} functions of $\set{J_{1}, J_{2}, \dotsc, J_{N}}$ on \emph{the open stratum} $\strata{H}$. Beware, however, that these rational expressions may \emph{not be unique}. However if some $\sigma_{i}$ can be written as $P_{1}/Q_{1}$ as well as $P_{2}/Q_{2}$ then $P_{1}Q_{2} - P_{2}Q_{1}$ belongs to the ideal generated by the $S_{j}$.

Because the solutions are rational, the fact that the field on which we work is real or complex does not matter at this level. Therefore, for each \emph{real} solution $(J_{1}, J_{2}, \dotsc, J_{N})$ of \eqref{eq:method_syzygies}, it corresponds a \emph{unique real} solution $(\sigma_{1}, \dotsc , \sigma_{r})$ given by \eqref{eq:method_solutions}. Nevertheless, we cannot take for granted that a real solution $(\sigma_{1}, \dotsc , \sigma_{r})$ of~\eqref{eq:method_solutions} corresponds to a real point $(x^{1}, \dotsc , x^{q})$ in $V^{H}$. For this, we need to compute an additional system of inequalities on the $\sigma_{i}$, or equivalently on the $J_{k}$, which permits to exclude complex solutions of
\begin{equation*}\label{eq:relative_invariants}
    \sigma_{i} = Q_{i}(x^{1}, \dotsc x^{q}), \quad i = 1, \dotsc , r, \quad Q_{i}\in\RR[X_{1},\dotsc,X_{q}]
\end{equation*}
However, in the present paper for the linear slices we consider, the only monodromy group that we encounter are\footnote{Where $\GS_{n}$ is the the group of permutations acting on $n$ elements (symmetric group).} $\triv$, $\GS_{2}$ and $\GS_{3}$ and the action is the standard one in appropriate coordinates. In these cases, and more generally when the monodromy group is isomorphic to $\GS_{n}$ and acts by permutation on the coordinates, the invariants $\sigma_{1}, \dotsc , \sigma_{n}$ are algebraically independent, $r=q=n$ and $x^{1}, \dotsc x^{n}$ are the roots of the polynomial
\begin{equation*}
    p(z) = z^{n} - \sigma_{1}z^{n-1} + \dotsb + (-1)^{n}\sigma_{n}
\end{equation*}
Therefore, the problem reduces to find conditions on the $\sigma_{i}$ that ensure that all the roots of $p$ are real. The solution is due to Hermite, \cite{Cos02}. He has proved that the number of distinct real roots of a real polynomial $p$ of degree $n$ is equal to the \emph{signature} of the Hankel  matrix $B(p) := \left( S_{i+j-2} \right)_{1 \leq i,j \leq n}$, where $S_k := \sum_{i=1}^{n} (x^i)^k$  is the power sum of the roots of $p$. In particular, $p$ has real roots if and only if $B(p)$ is non-negative.

In the forthcoming section this method will be applied to establish the stratification of a $n$-uple of quadratic forms. Besides being just an illustration, the results obtained will be of great interest for the symmetry detection of elasticity tensors  and used in \autoref{sec:elasticity_tensors} and \autoref{sec:stratifaction_ela}.

% ----------------------------------------------------------------
% ----------------------------------------------------------------

\section{The isotropic stratification of a $n$-uple of quadratic forms}
\label{sec:stratification_quadratic_forms}

In this section, we will study the isotropic stratification of the space
\begin{equation*}
    \bigoplus_{k=1}^{n} S^{2}(\RR^{3}) = n \HH^{0} \oplus n \HH^{2}.
\end{equation*}
For $n=2$, that is for a couple of quadratic forms, we get the lower order terms of the harmonic decomposition of $\Ela$ (c.f. \autoref{sec:elasticity_tensors}). Invariant conditions which characterize each symmetry class will be given.

% ----------------------------------------------------------------

\subsection{The isotropic stratification of one quadratic form}
\label{subsec:one_quadratic_form}

The lattice of isotropy classes for the tensorial representation of $\SO(3)$ on the space of symmetric 3-matrices, $S^{2}(\RR^{3})$, is a total order
\begin{equation*}
    [\DD_{2}] \preceq [\OO(2)] \preceq [\SO(3)] \, .
\end{equation*}
Since the least isotropy class, $[\DD_{2}]$, is not trivial, the fixed point set $S^{2}(\RR^{3})^{\DD_{2}}$ (which corresponds to the subspace of diagonal matrices) is a \emph{global} linear slice. Each $\SO(3)$-orbit intersects this slice (a symmetric matrix is diagonalizable in a direct orthogonal frame). This linear slice is of dimension 3 and degree 6 (see Appendix~\ref{sec:normalizers}). Generically, an orbit intersects this slice in exactly 6 points which correspond to the permutation of the eigenvalues $\lambda_1,\, \lambda_{2},\, \lambda_{3}$. The monodromy group $N(\DD_{2})/\DD_{2}\simeq \GS_{3}$ acts on $\set{\lambda_{1}, \lambda_{2}, \lambda_{3}}$ by permutation.
The stratification problem is reduced to the computation of the number of distinct eigenvalues, namely $(3,2,1)$. The invariant algebra of $(S^{2}(\RR^{3})^{\DD_{2}}, \GS_{3})$ is the free algebra generated by the 3 symmetric elementary functions $\sigma_1, \sigma_{2}, \sigma_{3}$. According to Hermite (see \cite{Cos02})  the characteristic polynomial of $a$ has  3 distinct real roots iff the quadratic form
\begin{equation*}
B(a) := \left(
  \begin{array}{ccc}
    3 & 2\,\sigma_{{1}} & \sigma_{{2}} \\
    2\,\sigma_{{1}} & 2\,{\sigma_{1}}^{2} - 2\,\sigma_{2} & \sigma_{1}\sigma_{2} - 3\,\sigma_{3} \\
    \sigma_{2} & \sigma_{1}\sigma_{2} - 3\,\sigma_{3} & - 2\,\sigma_{1} \sigma_{3} + {\sigma_{2}}^{2} \\
  \end{array}
\right)
\end{equation*}
is positive definite.
The three principal minors of $B$ are given by
\begin{equation*}
    \left\{
    \begin{array}{rcl}
	   \Delta_1 &=& 3, \\
	   \Delta_{2} &=& 2 \sigma_1^{2} - 6\sigma_{2}, \\
	\Delta_{3} &=& -27\,{\sigma_{3}}^{2}+\left( 18\,\sigma_{1}\,\sigma_{2}-4\,{\sigma_{1}}^{3}\right) \,\sigma_{3}-4\,{\sigma_{2}}^{3}+{\sigma_{1}}^{2}\,{\sigma_{2}}^{2}.
    \end{array}
    \right.
\end{equation*}
These minors are precisely those which have been introduced by Hermite to count the number of real roots of a real polynomial of degree $3$. $\Delta_{3}$ is the discriminant of the polynomial
\begin{equation*}
    \Delta_{3} = (\lambda_1 - \lambda_{2})^{2}  (\lambda_{2} - \lambda_{3})^{2}  (\lambda_1 - \lambda_{3})^{2} ,
\end{equation*}
whereas
\begin{equation*}
    \Delta_{2} = (\lambda_1 - \lambda_{2})^{2} + (\lambda_{2} - \lambda_{3})^{2} + (\lambda_1 - \lambda_{3})^{2}
\end{equation*}
vanishes when all the roots are equal. They give directly the classification we are looking for, which is summarized in table~\ref{tab:example2}.

\begin{table}[h]
\begin{equation*}
\begin{array}{|c|c|c|c|c|c|c|c|} \hline
H 		& N(H) 		& \Gamma^{H} 	& \card \Gamma^{H} 	& \text{strata } \Sigma_H 	& \dim V^{H} 	&  \dim \strata{H}/G &\dim \Sigma_H 	 \\ \hline \hline
\DD_{2}	& \octa	  & \GS_{3} & 6	& \Delta_{3} > 0  				 & 3	& 3	& 6 \\ \hline
\OO(2)	& \OO(2)  & \triv	& 1 & \Delta_{3}=0,\ \Delta_{2}>0  	 & 2 & 2	& 4 \\ \hline 	
\SO(3) 	& \SO(3)  & \triv& 1 & \Delta_{3}=0, \Delta_{2}=0	& 1 & 1 & 1 \\ \hline
\end{array}
\end{equation*}
\caption{Isotropy classes for the representation of $\SO(3)$ on symmetric 3-matrices.}
\label{tab:example2}
\end{table}

% ----------------------------------------------------------------

\subsection{The isotropic stratification of $n$ quadratic forms}
\label{subsec:n_quadratic_form}

In this subsection, we consider a $n$-uple of quadratic forms $(\mathbf{a}_{1}, \dotsc ,\mathbf{a}_{n})$. The lattice of isotropy classes is the following totally ordered set (see \cite{Oli11} for a general algorithm to compute isotropy classes)
\begin{equation*}\label{eq:nS2_isotropy_classes}
    [\triv] \preceq [\ZZ_{2}] \preceq [\DD_{2}] \preceq [\OO(2)] \preceq [\SO(3)] \,
\end{equation*}
The stratification is summarized in table~\ref{tab:example3}.
\begin{table}[h]
\begin{equation*}
\begin{array}{|c|c|c|c|c|c|c|} \hline
H 		& N(H) 		& \Gamma^{H}  & \card \Gamma^{H} & \dim V^{H} & \dim \strata{H}/G	& \dim \strata{H} \\ \hline \hline
\triv	& \SO(3)    & \SO(3)	  & \infty			 &  6n		  & 6n-3	            & 6n              \\ \hline
\ZZ_{2}	& \OO(2)	& \OO(2) 	  & \infty			 &  4n		  & 4n-1	            & 4n+2            \\ \hline
\DD_{2}	& \octa		& \GS_{3} 	  & 6				 &  3n		  & 3n	                & 3n+3            \\ \hline
\OO(2)	& \OO(2)	& \triv	  & 1 			     &  2n		  & 2n	                & 2n+2            \\ \hline
\SO(3) 	& \SO(3) 	& \triv	  & 1 			     &  n	      & n		            & n               \\ \hline
\end{array}
\end{equation*}
\caption{Isotropy classes for the representation of $\SO(3)$ on a $n$-uple of quadratic forms.}
\label{tab:example3}
\end{table}
To establish the bifurcation conditions we need to introduce the commutator of two quadratic forms. Given two quadratic forms $(\mathbf{a},\mathbf{b})$ on the euclidean space $\RR^{3}$, the commutator
\begin{equation*}
    [\mathbf{a},\mathbf{b}]:=\mathbf{a}\mathbf{b}-\mathbf{b}\mathbf{a}
\end{equation*}
is a skew-symmetric endomorphism $\Omega$, which can be represented, in the \emph{oriented} vector space $\RR^{3}$, by a vector $\omega$. We will denote by $\omega \mapsto \Omega := j(\omega)$ this linear isomorphism between vectors of $\RR^{3}$ and skew-symmetric second order tensors on $\RR^{3}$. Recall that for every $\mathbf{g}\in\SO(3)$ and $\omega\in\RR^{3}$, we have
\begin{equation*}
    j(\mathbf{g}\omega) = \mathbf{g}j(\omega)\mathbf{g}^{t}.
\end{equation*}
In particular
\begin{equation*}
    \omega (\mathbf{a},\mathbf{b}) := j^{-1}([\mathbf{a},\mathbf{b}])
\end{equation*}
is a covariant of degree 2 of the pair $(\mathbf{a},\mathbf{b})$.

\begin{thm}\label{thm:nS2_bifurcations}
The characterization of isotropy classes in $\bigoplus_{k=0}^{n}S^{2}(\RR^{3})$ is summarized below:
\begin{enumerate}

    \item A $n$-uple $(\mathbf{a}_{1}, \dotsc ,\mathbf{a}_{n})$ has at least isotropy $[\ZZ_{2}]$ iff there exists a non-vanishing vector $\omega\in \RR^{3}$ and real numbers $\lambda_{k,l}$ such that
        \begin{equation}\label{eq:nS2_bifurcation1}
        \begin{aligned}
          \omega(\mathbf{a_{k}},\mathbf{a_{l}}) & = \lambda_{k,l}\,\omega,& & 1 \le k,l \le n, \\
          \mathbf{a_{k}}\omega \wedge \, \omega & = 0,& & 1 \le k \le n ,
        \end{aligned}
        \end{equation}
        where $\wedge$ is the classical vector product of $\RR^{3}$.

    \item A $n$-uple $(\mathbf{a}_{1}, \dotsc ,\mathbf{a}_{n})$ has at least isotropy $[\DD_{2}]$ iff
        \begin{equation}\label{eq:nS2_bifurcation2}
            \omega(\mathbf{a_{k}},\mathbf{a_{l}}) = 0, \qquad 1 \le k,l \le n .
        \end{equation}

    \item A $n$-uple $(\mathbf{a}_{1}, \dotsc ,\mathbf{a}_{n})$ has at least isotropy $[\OO(2)]$ iff
        \begin{equation}\label{eq:nS2_bifurcation3}
        \begin{aligned}
          25 \left[\tr(a_{k}^{2})\right]^{3} & = 54 \left[\tr(a_{k}^{3})\right]^{2},&  & 1 \le k \le n , \\
          \left[\tr(a_{k}a_{l})\right]^{2} & = \left[\tr(a_{k}^{2})\right]\left[\tr(a_{l}^{2})\right],&  & 1 \le k,l \le n ,
        \end{aligned}
        \end{equation}
        where $a_{k}=\dev(\mathbf{a}_{k})$ is the deviatoric part of $\mathbf{a}_{k}$.

    \item A $n$-uple $(\mathbf{a}_{1}, \dotsc ,\mathbf{a}_{n})$ has isotropy $[\SO(3)]$ iff
        \begin{equation}\label{eq:nS2_bifurcation4}
            a_{k} = 0, \qquad 1 \le k \le n ,
        \end{equation}
        where $a_{k}=\dev(\mathbf{a}_{k})$ is the deviatoric part of $\mathbf{a}_{k}$.
\end{enumerate}
\end{thm}

\begin{proof}
(1) If a $n$-uple $(\mathbf{a}_{1}, \dotsc ,\mathbf{a}_{n})$ has at least isotropy $[\ZZ_{2}]$, its elements share a common principal axis (the axis of the non trivial rotation $\sigma$ in $\ZZ_{2}$). Let $\omega$ be a unit vector parallel to this axis. Then $\omega$ belongs to the kernel of the commutator of each pair $(\mathbf{a_{k}},\mathbf{a_{l}})$ and thus $\omega(\mathbf{a_{k}},\mathbf{a_{l}})$ is proportional to $\omega$. Relations~\eqref{eq:nS2_bifurcation1} are a consequence of these facts. Conversely, if relations~\eqref{eq:nS2_bifurcation1} are satisfied, then $\omega$ is a common eigenvector of the $n$-uple $(\mathbf{a}_{1}, \dotsc ,\mathbf{a}_{n})$. In that case, the isotropy group of the pair $(\mathbf{a},\mathbf{b})$ is at least in the class $[\ZZ_{2}]$.

(2) If a $n$-uple $(\mathbf{a}_{1}, \dotsc ,\mathbf{a}_{n})$ has at least isotropy $[\DD_{2}]$, then there exists a $g\in\SO(3)$ which brings all the $\mathbf{a}_{k}$ to a diagonal form. Thus they commute. Conversely, if the $\mathbf{a}_{k}$ commute, they can be diagonalized simultaneously and the isotropy group of the $n$-uple $(\mathbf{a}_{1}, \dotsc ,\mathbf{a}_{n})$ is at least in the class $[\DD_{2}]$.

(3) A $n$-uple $(\mathbf{a}_{1}, \dotsc ,\mathbf{a}_{n})$ has at least isotropy $[\OO(2)]$ iff all the $\mathbf{a}_{k}$ can be diagonalized simultaneously and, moreover, iff their respective deviator has (in a well chosen basis) a diagonal form proportional to $\mathrm{diag}(1,1,-2)$. The first relation of equations~\eqref{eq:nS2_bifurcation3} expresses that the deviator of $\mathbf{a}_{k}$ has a double eigenvalue, the second relation expresses that the deviators $a_{k}$ and $a_{l}$ are proportional (Cauchy-Schwarz relation). Conversely, if equations~\eqref{eq:nS2_bifurcation3} are satisfied, then all the deviators $a_{k}$ commute and each of them has a double eigenvalue. It is then possible to find a basis in which all the $a_{k}$ are proportional to $\mathrm{diag}(1,1,-2)$. Therefore the $n$-uple $(\mathbf{a}_{1}, \dotsc ,\mathbf{a}_{n})$ has at least isotropy $[\OO(2)]$.

(4) The last assertion is trivial.
\end{proof}

% ----------------------------------------------------------------
% ----------------------------------------------------------------

\section{The space of elasticity tensors}
\label{sec:elasticity_tensors}

% ----------------------------------------------------------------

\subsection{The harmonic decomposition of elasticity tensors}
\label{subsec:harmonic_decomposition}

To establish the isotropic stratification of a tensor space, the first step is to decompose it into a direct sum of elementary pieces. For the $\SO(3)$-action, this decomposition is known as the harmonic decomposition. The $\SO(3)$-irreducible pieces are symmetric and traceless tensors (i.e. harmonic tensors). The space of elasticity tensors admits the following harmonic decomposition \cite{Bac70,Bae93,FV96,FV06}:
\begin{equation} \label{eq:elastic_harmonic_decomposition}
    \Ela   \simeq 2 \HH^{0} \oplus 2 \HH^{2}  \oplus \HH^{4}
\end{equation}
Therefore, each $\mathbf{C} \in \Ela$ can be written as  $\mathbf{C} = (\lambda, \mu, a, b, D)$ where $\lambda, \mu \in \HH^{0}$,
$a, b \in \HH^{2}$ and $D \in \HH^{4}$ and such that for all $g\in \SO(3)$
\begin{equation*}
    \bar{\mathbf{C}} = \rho(g)\cdot \mathbf{C} \Longleftrightarrow \bar{\lambda} = \lambda,\, \bar{\mu} = \mu,\, \bar{a} = \rho_{2}(g)\cdot a,\, \bar{b} = \rho_{2}(g)\cdot b,\, \bar{D} = \rho_{4}(g)\cdot D
\end{equation*}
Elements of the set $(\lambda, \mu, a, b, D)$ are covariants (c.f. \autoref{subsec:covariants}) to  $\mathbf{C}$ : $(\lambda, \mu)$ are invariants,  $(a, b)$ are $\rho_{2}$-covariants, and $D$ is $\rho_{4}$-covariant.
The explicit harmonic decomposition is well known, given an orthonormal frame $(e_{1}, e_{2}, e_{3})$ of the Euclidean space, we get \cite{BKO94}:
\begin{equation}\label{eq:decomposition_boehler}
    C_{ijkl}  = \lambda \delta_{ij} \delta_{kl} + \mu(\delta_{ik} \delta_{jl} + \delta_{il} \delta_{jk})+ \delta_{ij} a_{kl} + \delta_{kl} a_{ij} + \delta_{ik} b_{jl} + \delta_{jl} b_{ik} + \delta_{il} b_{jk} + \delta_{jk} b_{il} + D_{ijkl}
\end{equation}
where the metric is written $q_{ij} := \delta_{ij}$. This formula can be inverted to obtain the $5$ harmonic tensors $(\lambda, \mu, a, b, D)$ from $\mathbf{C}$.
On $\Ela=S^{2}S^{2}(\RR^{3})$, there are only two different traces:
\begin{equation*}
        d_{ij}=(\tr_{12} \mathbf{C})_{ij} := \sum_{k=1}^{3} C_{kkij}, \quad
        v_{ij}=(\tr_{13} \mathbf{C})_{ij} := \sum_{k=1}^{3} C_{kikj}.
\end{equation*}
where $\mathbf{d}$ and $\mathbf{v}$ are known  as, respectively, the \emph{dilatation} tensor and the \emph{Voigt} tensor \cite{CM87,Cow89}.
Starting with \eqref{eq:decomposition_boehler} we get
\begin{equation}\label{eq:Dilat_Voigt_Tensors}
    \mathbf{d}  = (3\lambda + 2 \mu)\mathbf{q} + 3a + 4b, \quad
    \mathbf{v}  = (\lambda + 4 \mu)\mathbf{q} + 2a + 5b.
\end{equation}
Taking the traces of each equation, one obtains:
\begin{equation*}
    \tr(\mathbf{d})  =  9\lambda + 6 \mu , \quad
    \tr(\mathbf{v})  =  3\lambda + 12 \mu
\end{equation*}
and, finally:
\begin{equation*}\label{eq:decomposition_boehler_inverse}
  \begin{aligned}
   \lambda   & = \frac{1}{15}( 2  \tr(\mathbf{d})   -    \tr(\mathbf{v})),\quad  &\mu      & =  \frac{1}{30}(-   \tr(\mathbf{d})   + 3    \tr(\mathbf{v})) \\
   a          & =  \frac{1}{7}( 5\dev(\mathbf{d}) - 4 \dev(\mathbf{v})) ,\quad   &b          & = \frac{1}{7} ( -2\dev(\mathbf{d})   + 3   \dev(\mathbf{v}))
  \end{aligned}
\end{equation*}
where we write $\dev(\mathbf{a}) :=\mathbf{a} - \frac{1}{3}\tr(\mathbf{a})\, \mathbf{q}$ the \emph{deviatoric} part of a 2nd-order tensor.

The isotropy classes of elasticity tensors have been computed in \cite{FV96,BBS04}. There are exactly 8 classes: the \emph{isotropic} class $[\SO(3)]$, the \emph{cubic} class $[\octa]$, the \emph{transversely isotropic} class $[\OO(2)]$, the trigonal class $[\DD_{3}]$, the tetragonal class $[\DD_{4}]$, the orthotropic class $[\DD_{2}]$, the monoclinic class $[\ZZ_{2}]$ and the totally anisotropic class $[\triv]$ (also know as the triclinic class). We refer to \autoref{sec:normalizers} for the precise definitions of these groups. This stratification is a consequence of the index symmetries of elasticity tensors. It is not a general result that encompass all kind of 4th-order tensors. For example, the photo-elasticity tensors space which is a  4th-order tensor space solely endowed with minor symmetries is divided into 12 symmetry classes~\cite{FV97}.

% ----------------------------------------------------------------

\subsection{The isotropic stratification of $\Ela$}
\label{subsec:isotropic_stratification_ela}

The stratification of $\Ela$ is provided in the table \ref{tab:Ela}.
\begin{table}[h!]
$$
\begin{array}{|c|c|c|c|c|c|c|c|} \hline
H 		& N(H) 		& \Gamma^{H} & \card \Gamma^{H} & \dim V^{H} & \dim \strata{H}/G & \dim \strata{H} \\ \hline \hline
\triv	& \SO(3) 	& \SO(3) 	 & \infty		    & 21		 & 18	             & 21			   \\ \hline
\ZZ_{2}	& \OO(2)	& \OO(2)	 & \infty		    & 13		 & 12		         & 15			   \\ \hline
\DD_{2}	& \octa		& \GS_{3} 	 & 6				& 9			 & 9	             & 12			   \\ \hline
\DD_{3}	& \DD_{6}	& \GS_{2} 	 & 2				& 6		     & 6		         & 9			   \\ \hline
\DD_{4}	& \DD_{8}	& \GS_{2} 	 & 2				& 6			 & 6	             & 9			   \\ \hline
\OO(2)	& \OO(2)	& \triv		 & 1 			    & 5		     & 5		         & 7 			   \\ \hline 	
\octa	& \octa		& \triv		 & 1 			    & 3			 & 3	             & 6			   \\ \hline 	
\SO(3) 	& \SO(3) 	& \triv 	 & 1 			    & 2	         & 2		         & 2			   \\ \hline
\end{array}
$$
\caption{Isotropy classes for $\Ela$.}
\label{tab:Ela}
\end{table}
The dimensions of $\strata{H}$ and $\strata{H}/G$ are obtained using the formulas \eqref{eq:quotient_strata_dimenion} of \autoref{subsec:strata_dimension}. And $\dim V^{H}$ was computed using the \emph{trace formula}~\eqref{eq:trace_formula} and the explicit formulas provided \autoref{sec:TrFrm}.
Table \ref{tab:Ela} sums up the different dimensions that are classically associated to anisotropic elasticity tensors in the literature. All these dimensions are linked by the formula
\begin{equation*}\label{eq:stratum_dimension2}
	\dim \strata{H} = \dim \strata{H}/G + \dim \Gamma^{H} + e_{H},
\end{equation*}
where $\strata{H}/G$ is the orbit space for tensors in isotropy class $[H]$, $\Gamma^{H} = N(H)/H$ is the monodromy group and $e_{H} = \dim G/N(H)$. Notice that $e_{H}\simeq [H]$ is the set of distinct subgroups of $G$ which are conjugate to $H$.
These results of elasticity enlighten the physical contents of the geometric concepts introduced in \autoref{sec:orbit_spaces}.
\begin{description}
\item[$V^{H}$ ] The fixed-point set $V^{H}$ is the vector space of tensors having at least $H$ as isotropy group. Its dimension equals the dimension of the matrix space used to represent them. The space $V^{H}$ can be constructed for any closed subgroup $H$, not only for an isotropy subgroup. In elasticity, for instance, fixed-point sets have been constructed for $\ZZ_{3}$, $\ZZ_{4}$ (which are not isotropy subgroups). This is why it was believed for a long time that there exists $10$ type of elasticity. The construction of such matrix spaces is detailed, for example, in \cite{CM90} for classical elasticity, and extended to second order elasticity in \cite{ABB09}. It worthwhile emphasizing the fact that for a non-isotropy subgroup, $V^{H}$ is solely a fixed-point set and in no way a slice. It is geometrically meaningless for the description of the stratification, that is the geometry of the orbit space.

\item[$\strata{H}$ ] The set of tensors which symmetry group is conjugate to $H$ forms a stratum, or a symmetry class according to Forte and Vianello \cite{FV96}. We got the following dimension relation $\dim\strata{H}=\dim V^{H} + e_{H}$ where $e_{H} = \dim \SO(3)/N(H)$ is the number of Euler angles needed to define the appropriate coordinate system in which the elastic tensor has the minimum number of elastic constants \cite{Nor89}. The strata dimension corresponds to what Norris, for example, calls the number of independent parameters characterizing a tensor.

\item[$\strata{H}/G$ ] This is the orbit space of $\strata{H}$, i.e. each point of this space represents the $G$-orbit of a tensor which symmetry group is conjugate to $H$.
The intrinsic properties of a tensor are characterized by its orbit. In a slice $V^{H}$ the image of the orbit is given by monodromy group, and we have the dimensional relation $\dim V^{H}=\dim\strata{H}/G+\dim\Gamma^{H}$.
Therefore when $\Gamma^{H}$ is discrete, $\strata{H}/G$ and $V^{H}$ have the same dimension, this is the situation for most of elasticity symmetry classes. At the opposite, when $\Gamma^{H}$ is continuous, the dimension of the orbit is strictly smaller than the one of the slice, in elasticity this case occurs for $\triv,\ZZ_{2}$.
\end{description}

\begin{rem}
It is known for long time that for the symmetry classes $\triv,\ZZ_{2}$, there is no \emph{canonical way} to obtain a basis in which the number of non-zero components is minimal. More precisely, an element of the orbit of such a tensor in an appropriate fixed point set $V^{H}$ does not have necessarily the minimal number of non-zero components and thus, such a basis has to be computed \emph{arbitrarily}. This is due to the fact that in these cases, $\Gamma^{H}$ is continuous.  For the trivial class, $[\triv]$, $V^{H} \simeq V$ and there is no natural way to obtain a normal form. The monodromy group $\Gamma^{H}\simeq \SO(3)$ and there are three arbitrary degree of freedom to choose a \emph{nice basis}. For the other cases, $V^{H}$ is smaller than $V$ but the number of non-zero components of a tensor in $V^{H}$ can usually be lowered arbitrarily. The monodromy $\Gamma^{H}\simeq \OO(2)$ and there is one arbitrary degree of freedom to lower the number of non-zero components. In each cases, it appears clearly that one can properly choose a basis in which those angular parameters equal zeros, such a basis is non necessarily unique. In such a basis the number of independent parameters is lowered and we obtain the \emph{minimum number of elastic constants}, as discussed, for example, by Norris \cite{Nor89}. This minimum number of elastic constants is equal to the orbit dimension.
\end{rem}
The results appearing in table \ref{tab:Ela} are well-known in elasticity, but their constructions are usually ad-hoc and do not stand on rigorous mathematical foundations. Therefore, we hope to provide such a firm geometric construction. Conversely, as the geometrical concepts we use are abstract, these well-known mechanical results help to illustrate their physical contents.
As it has been said before the space of elasticity tensors can be decomposed in the following way:
\begin{equation}
    \Ela  \simeq 2 \HH^{0} \oplus 2 \HH^{2} \oplus \HH^{4}
\end{equation}
The \autoref{sec:stratification_H4} will be devoted to the stratification of $\HH^{4}$. The stratification of $2 \HH^{0} \oplus 2 \HH^{2}$ is a consequence of the results obtained in the \autoref{subsec:n_quadratic_form} in the case $n=2$. In such a case a direct link between the bifurcation conditions given \eqref{thm:nS2_bifurcations} and the ones appearing in the Cowin-Mehrabadi theorem can be drawn.

% ----------------------------------------------------------------

\subsection{The Cowin-Mehrabadi condition}
\label{subsec:cowin_mehrabadi}

The information contained in the Voigt and dilatational tensors \eqref{eq:Dilat_Voigt_Tensors} solely concerns the covariants  $(\lambda, \mu, a, b)$.  As a consequence any statement on $\mathbf{d}$ and $\mathbf{v}$ can be translated in terms of $(\lambda, \mu, a, b)$, and conversely.

The condition \eqref{eq:nS2_bifurcation1} for $(\mathbf{a},\mathbf{b})$ to have at least isotropy $[\ZZ_{2}]$ means that $\omega=[\mathbf{a},\mathbf{b}]$ is a common eigenvector of $(\mathbf{a},\mathbf{b})$. If we note $n=\omega$, it is a straightforward computation to show that this condition on $(a,b)$ is equivalent to
\begin{align*}\label{eq:Cond_dv}
   d_{ij}n_{j}=\delta n_{i}, \quad \text{and} \quad
   v_{ij}n_{j}=\nu n_{i}.\\
\end{align*}
These conditions are nothing but the two first ones of the Cowin-Mehrabadi (CM) theorem  for the existence of a symmetry plane.
Let summarize this theorem as it is formulated in \cite{Tin03}.
\begin{thm}[Cowin-Mehrabadi]
Let $n$ be a unit vector normal to a plane of material symmetry, if  $m$ is any vector on the symmetry plane a set of necessary and sufficient conditions for $n$ to be normal to the symmetry plane is
\begin{align*}
   &C_{ijkk}n_{j}=(C_{pqtt}n_{p}n_{q})n_{i},\\
   &C_{isks}n_{k}=(C_{pqrq}n_{p}n_{r})n_{i},\\
   &C_{ijks}n_{j}n_{s}n_{k}=(C_{pqrt}n_{p}n_{q}n_{r}n_{t})n_{i},\\
   &C_{ijks}m_{j}m_{s}n_{k}=(C_{pqrt}n_{p}m_{q}n_{r}m_{t})n_{i}.\\
\end{align*}
\end{thm}
The two first conditions of this theorem can be rephrased in terms of the second order covariants $(a,b)$ of the harmonic decomposition. The bifurcation condition \eqref{eq:nS2_bifurcation1} provides a method to find a good candidate to be the normal of a symmetry plane $\omega=[a,b]$.
In the same spirit the third and fourth conditions are indeed conditions on $D$. Therefore the theorem can be reformulated in term of harmonic covariant using the condition \eqref{eq:nS2_bifurcation1}:
\begin{thm}[Cowin-Mehrabadi]
Let $\mathbf{C}=(\lambda, \mu, a, b, D)$ be an elasticity tensor. Let suppose there exist a  non null-vector $\omega=[a,b]$, this vector is a vector normal to a plane of material symmetry, if there exists $\omega^{\bot}$ verifying $\omega^{\bot}\cdot\omega^{\bot}=0$ and such that the following conditions are satisfied
\begin{align*}
   &\omega\wedge(D:\omega\otimes\omega)\omega=0,\\
   &\omega\wedge(D:\omega^\bot\otimes\omega^\bot)\omega=0.\\
\end{align*}
\end{thm}

% ----------------------------------------------------------------
% ----------------------------------------------------------------

\section{The isotropic stratification of $\HH^{4}$}
\label{sec:stratification_H4}

In this section,  the isotropic stratification of $\HH^{4}$ will be established. This space is isomorphic to the higher order component of the harmonic decomposition of $\Ela$. Invariant conditions characterizing each finite monodromy symmetry class will be given.  These results provided new sets of necessary conditions to identify higher order symmetry classes, and will be used  in \autoref{sec:stratifaction_ela}.

\subsection{Lattice of isotropy}

The isotropy classes of the 9-D space $\HH^{4}$ are the same as the 8 classes of $\Ela$. The corresponding lattice\footnote{It should be noticed that it is not possible to \emph{realize} this lattice of conjugacy classes as a lattice of subgroups (for the inclusion relation). In  particular, if each conjugacy class in figure~\ref{fig:lattice1} is replaced by the corresponding subgroup defined in Appendix~\ref{sec:normalizers}, then almost all the arrows corresponds to an inclusion relation but not the arrow from $(4)$ to $(2)$ : $\DD_{3}$ is not a subgroup of $\octa$ but is conjugate to a subgroup of $\octa$.} is illustrated in figure~\ref{fig:lattice1}. With the convention that a group at the starting point of an arrow is conjugate to a subgroup of the group pointed by the arrow.

\begin{figure}[h]
  \centering
  \includegraphics{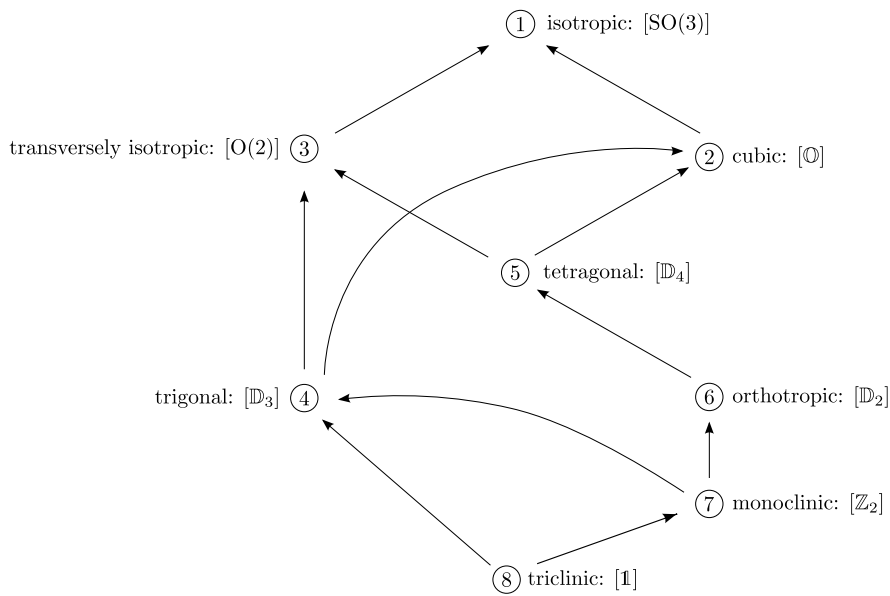}\\
  \caption{The lattice of isotropy classes for $\HH^{4}$.}
  \label{fig:lattice1}
\end{figure}

Contrary to what happens with quadratic forms\footnote{It is quite remarkable that for quadratic forms in every dimension $n$, the generic stratum has \emph{non trivial isotropy} and that the corresponding linear slice has finite monodromy $\GS_{n}$. This is the main reason why every orbit has a \emph{natural normal form}, namely a diagonal representation (up to the finite monodromy group $\GS_{n}$). In comparison, the main difficulty in linear elasticity is the fact that for a generic elasticity tensor, the isotropy group is trivial and there is no natural normal form.}, there is no nice description of the isotropic stratification. The aim of this section is to describe the strata which have \emph{finite} monodromy group $\Gamma^{H}$: the \emph{cubic} class $[\octa]$, the \emph{transversely isotropic} class $[\OO(2)]$, the trigonal class $[\DD_{3}]$, the tetragonal class $[\DD_{4}]$ and the orthotropic class $[\DD_{2}]$.
For each of them, explicit results will be provided. The algebraic procedure used to obtain their solutions was explained in~\autoref{subsec:implicitization}.

\begin{table}[h]
$$
\begin{array}{|c|c|c|c|c|c|c|c|} \hline
H 		& N(H) 		& \Gamma^{H} 	& \card \Gamma^{H} & \dim V^{H} & \dim \strata{H}/G & \dim \strata{H} \\ \hline \hline
\triv	& \SO(3) 	& \SO(3) 	& \infty		& 9		&6	& 9				 \\ \hline
\ZZ_{2}	& \OO(2)	& \OO(2)	& \infty		& 5	&4		& 7 			 \\ \hline
\DD_{2}	& \octa		& \GS_{3} 	& 6				& 3		&3	& 6				 \\ \hline
\DD_{3}	& \DD_{6}	& \GS_{2} 	& 2				& 2	&2		& 5				 \\ \hline
\DD_{4}	& \DD_{8}	& \GS_{2} 	& 2				& 2		&2	& 5				 \\ \hline
\OO(2)	& \OO(2)	& \triv		& 1 			& 1	&1		& 3 			 \\ \hline 	
\octa	& \octa		& \triv		& 1 			& 1		&1	& 4				 \\ \hline 	
\SO(3) 	& \SO(3) 	& \triv 	& 1 			& 0&	0		 & 0				\\ \hline
\end{array}
$$
\caption{Isotropy classes for the tensorial representation of $\SO(3)$ on $\HH^{4} $.}
\label{tab:H4}
\end{table}

% ----------------------------------------------------------------

\subsection{Global parametrization}
\label{subsec:global_parametrization}

On an euclidean vector space, any bilinear form can be represented by a symmetric linear operator. This applies in particular to an elasticity tensor $\mathbf{C}$ which may be considered as a quadratic form on $S^{2}(\RR^{3})$, where the Euclidean structure is induced by the standard inner product on $\RR^{3}$. It is therefore customary to represent an elasticity tensor  $\mathbf{C}$ by its corresponding symmetric linear operator $\underline{\mathbf{C}}$ on $S^{2}(\RR^{3})$ (relatively to the canonical scalar product $\sprod{a}{b} = \tr(ab)$). The corresponding matrix representation\cite{Kel56,CM90} is given by
\begin{equation*}
    \underline{\mathbf{C}} =\left(\begin {array}{cccccc}
        c_{11} & c_{12} & c_{13} & \sqrt {2}c_{14} & \sqrt {2}c_{15} & \sqrt {2}c_{16} \\
        c_{12} & c_{22} & c_{23} & \sqrt {2}c_{24} & \sqrt {2}c_{25} & \sqrt {2}c_{26} \\
        c_{13} & c_{23} & c_{33} & \sqrt {2}c_{34} & \sqrt {2}c_{35} & \sqrt {2}c_{36} \\
        \sqrt {2}c_{14} & \sqrt {2}c_{24} & \sqrt {2}c_{34} & 2\ c_{44} & 2\ c_{45} & 2\ c_{46} \\
        \sqrt {2}c_{15} & \sqrt {2}c_{25} & \sqrt {2}c_{35} & 2\ c_{45} & 2\ c_{55} & 2\ c_{56} \\
        \sqrt {2}c_{16} & \sqrt {2}c_{26} & \sqrt {2}c_{36} & 2\ c_{46} & 2\ c_{56} & 2\ c_{66}
    \end {array} \right)
\end{equation*}
where $c_{mn}$ are the components of the elasticity tensor in an orthonormal frame.
We use the conventional rule  to recode a pair of indices $(i,j)$ ($i,j = 1,2,3$) by a integer $m = \alpha(i,j)$ ($m = 1,2,\dotsc 6$), where $\alpha(i,i) := i$ for $1 \leq i\leq 3$ and $\alpha(i,j):= 9-(i+j)$ for $i\neq j$.

Given a general tensor $D \in \HH^{4} $, we denote by $\underline{D}$ the corresponding symmetric linear operator. In the canonical basis of $\RR^{3}$, $\underline{D}$ has the following matrix representation, depending on 9 real parameters:
\begin{equation*}\label{eq:generic_case}
  \underline{D} = \begin{pmatrix}
         D_{11} & \sqrt{2}\,D_{12} \cr
         \sqrt{2}\, D_{12}^{T} & 2\,D_{22} \cr
    \end{pmatrix}
\end{equation*}
where $^T$ represents the transposition,
\begin{equation*}
D_{11}= \begin{pmatrix}
  -h_{{9}}-h_{{8}}&h_{{9}}&h_{{8}}
\\ h_{{9}}&-h_{{9}}-h_{{7}}&h_{{7}}
\\ h_{{8}}&h_{{7}}&-h_{{8}}-h_{{7}}
\end{pmatrix}
    \qquad
D_{12}= \begin{pmatrix}
  -h_{{5}}-h_{{6}}&h_{{2}}&h_{{1}}
\\ h_{{5}}&-h_{{2}}-h_{{4}}&h_{{3}}
\\ h_{{6}}&h_{{4}}&-h_{{1}}-h_{{3}}
\end{pmatrix} ,
\end{equation*}
and
\begin{equation*}
D_{22}= \begin{pmatrix} h_{{7}}&-h_{{1}}-h_{{3}}&-h_{{2}}-h_{{4}}
\\ -h_{{1}}-h_{{3}}&h_{{8}}&-h_{{5}}-h_{{6}}
\\ -h_{{2}}-h_{{4}}&-h_{{5}}-h_{{6}}&h_{{9}}
\end{pmatrix}
\end{equation*}

Using this parametrization, $\underline{D}$ has the following invariance properties:

\begin{menumerate}
  \item $\ZZ_{2}$-invariant iif $h_{2} = h_4 = h_5 = h_6 = 0$\,,
  \item $\DD_{2}$-invariant iif $h_k = 0$ pour $1\leq k \leq 6$\,,
  \item $\DD_{3}$-invariant iif $h_1=h_{2}=h_{3}=h_4=h_6=0$ et $h_7=h_8=-4h_9$\,,
  \item $\DD_4$-invariant iif $h_k = 0$ pour $1\leq k \leq 6$ et $h_7=h_8$\,,
  \item $\OO(2)$-invariant iff $h_k = 0$ pour $1\leq k \leq 6$ et $h_7=h_8=-4h_9$\,,
  \item $\octa$-invariant iff $h_k = 0$ pour $1\leq k \leq 6$ et $h_7=h_8=h_9$\,,
  \item $\SO(3)$-invariant iff $h_k = 0$ pour $1\leq k \leq 9$\,.
\end{menumerate}

In the following subsections, the fundamental invariants of $4$-th order harmonic tensors will be used. As these results are not well-known we will, in the next lines, summarize their expressions as obtained by Boelher and al.\cite{BKO94}. For technical details and historical considerations, we therefore refer to this publication, and the references therein.

\begin{prop}
Let $D \in \HH^{4}$. The 2nd-order tensors $\mathbf{d}_{2}, \dotsc , \mathbf{d}_{10}$:
\begin{equation} \label{eq:boehler_invariants}
\begin{array} {lll}
     \mathbf{d}_{2} = \tr_{13}(D^{2})     &   \mathbf{d}_{3} = \tr_{13}(D^{3})    & \mathbf{d}_{4} = \mathbf{d}_{2}^{2} \\
     \mathbf{d}_{5} = \mathbf{d}_{2} D \mathbf{d}_{2}     &   \mathbf{d}_{6} = \mathbf{d}_{2}^{3}          & \mathbf{d}_{7} = \mathbf{d}_{2}^{2} D \mathbf{d}_{2} \\
     \mathbf{d}_{8} = \mathbf{d}_{2}^{2} D^{2} \mathbf{d}_{2} & \mathbf{d}_{9} = \mathbf{d}_{2}^{2} D \mathbf{d}_{2}^{2}  & \mathbf{d}_{10} = \mathbf{d}_{2}^{2} D^{2} \mathbf{d}_{2}^{2}
\end{array}
\end{equation}
are covariant to $D$. % as well as the traceless quadratic form $d_{s} = D\mathbf{d}_{2}$.
The nine fundamental invariants are defined by:
\begin{equation*}
    J_{k} := \tr(\mathbf{d}_{k}), \qquad k = 2, \dotsc ,10.
\end{equation*}
\end{prop}

The first 6 invariants $J_{2}, \dotsc , J_{7}$ are algebraically independent. The last 3 ones $J_{8},J_{9},J_{10}$ are linked to the formers by polynomial relations. These \emph{fundamental syzygies} were computed in~\cite{Shi67}.

\begin{rem}
Notice that the first invariant $J_{2}(D) = \sprod{D}{D}$ is the squared norm of $D$, it corresponds to the squared \emph{Frobenius norm} of $\underline{D}$. In particular, $J_{2}(D)=0$ if and only if $D=0$.
\end{rem}

% ----------------------------------------------------------------

\subsection{Cubic symmetry ($[\octa]$)}
\label{subsec:cubic}

A $\octa$-invariant tensor $D \in \HH^{4} $ has the following matrix representation:
\begin{equation}\label{eq:cubic_case}
  \underline{D} = \left(
    \begin {array}{cccccc}
    8\,\delta & -4\,\delta & -4\,\delta & 0 & 0 & 0 \\
    -4\,\delta & 8\,\delta & -4\,\delta & 0 & 0 & 0 \\
    -4\,\delta & -4\,\delta & 8\,\delta & 0 & 0 & 0 \\
    0 & 0 & 0& -8\,\delta & 0 & 0 \\
    0 & 0 & 0 & 0 & -8\,\delta & 0 \\
    0 & 0 & 0 & 0 & 0 & -8\,\delta
    \end {array}
 \right)
\end{equation}
where $\delta \in \RR$. The evaluation of the invariants~\eqref{eq:boehler_invariants} on this slice gives the following parametric system:
\begin{align*}
    J_{2} & = 480\,{\delta}^{2}, &  J_{3} & = 1920\,{\delta}^{3}, &  J_{4} & = 76800\,{\delta}^{4}, \\
    J_{5} & =0, & J_{6} & = 12288000\,{\delta}^{6}, & J_{7} & =0, \\
    J_{8} & =0, & J_{9} & = 0,  & J_{10} & = 0.
\end{align*}
We notice that the parameter $4 \delta = J_{3}/J_{2}$ is a rational invariant. The linear slice corresponding to the octahedral group is of degree 1, which was already known from the algebraic argument explained in \autoref{sec:normalizers}.

\begin{prop}\label{prop:cubic}
A harmonic tensor $D \in \HH^{4} $ is in the closed stratum $\cstrata{\octa}$ if and only if the invariants $J_{2}(D) \cdots J_{10}(D)$ satisfy the following polynomial relations
\begin{equation}\label{eq:cubic_syzigies}
\begin{aligned}
  3\,J_{4} & = {J_{2}}^{2}, & J_{5} & = 0, & 30\,{J_{3}}^{2} & = {J_{2}}^{3}, & 9\,J_{6} & = {J_{2}}^{3}, \\
  J_{7} & = 0,  & J_{8} & = 0, & J_{9} & =0, & J_{10} & = 0.
\end{aligned}
\end{equation}
$D$ is in the cubic class $[\octa]$ if moreover $J_{2}(D) \ne 0$. In that case, it admits the normal form~\eqref{eq:cubic_case} where $4 \delta := J_{3}(D)/J_{2}(D)$.
\end{prop}

We observe that $\forall J_{2},J_{3} \in\RR,\delta\in\RR$. Therefore we do not need to add any inequality to the syzygy system~\eqref{eq:cubic_syzigies}.

\begin{rem}
All 2nd-order tensors in $\HH^{2}$ covariant to $D$ vanish.
\end{rem}

% ----------------------------------------------------------------

\subsection{Transversely isotropic symmetry ($[\OO(2)]$)}
\label{subsec:transversely_isotropic}

A $\OO(2)$-invariant tensor $D \in \HH^{4} $ has the following matrix representation:
\begin{equation}\label{eq:transversly_isotropic_case}
  \underline{D} = \left(
   \begin {array}{cccccc}
    3\,\delta & \delta & -4\,\delta & 0 & 0 & 0 \\
    \delta & 3\,\delta & -4\,\delta & 0 & 0 & 0 \\
    -4\,\delta & -4\,\delta & 8\,\delta & 0 & 0 & 0 \\
    0 & 0 & 0 & -8\,\delta & 0 & 0 \\
    0 & 0 & 0 & 0 & -8\,\delta & 0 \\
    0 & 0 & 0 & 0 & 0 & 2\,\delta
   \end {array}
 \right)
\end{equation}
where $\delta \in \RR$. The evaluation of the invariants~\eqref{eq:boehler_invariants} on this slice gives:
\begin{align*}
  J_{2} & = 280\,{\delta}^{2}, & J_{3} & = 720\,{\delta}^{3}, & J_{4} & = 32800\,{\delta}^{4}, \\
  J_{5} & = 80000\,{\delta}^{5}, & J_{6} & = 4528000\,{\delta}^{6}, & J_{7} & = 17600000\,{\delta}^{7}, \\
  J_{8} & = 211200000\,{\delta}^{8}, & J_{9} & = 3872000000\,{\delta}^{9}, & J_{10} & = 46464000000\,{\delta}^{10}
\end{align*}
The parameter $\delta = {7\,J_{3}}/{18\,J_{2}}$ is a rational invariant. The linear slice corresponding to the group $\OO(2)$ is of degree 1, which was already known (see \autoref{sec:normalizers}).

\begin{prop}\label{prop:transversely_isotropic}
A harmonic tensor $D \in \HH^{4} $ is in the closed stratum $\cstrata{\OO(2)}$ if and only if
the invariants $J_{2}(D) \cdots J_{10}(D)$ satisfy the following polynomial relations
\begin{equation}\label{eq:transversly_isotropic_syzigies}
\begin{aligned}
    98\,J_{4} & = 41\,{J_{2}}^{2}, & 63\,J_{5} & = 25\,J_{3}J_{2}, & 3430\,{J_{3}}^{2} & = 81\,{J_{2}}^{3}, & 1372\,J_{6} & = 283\,{J_{2}}^{3}, \\
    882\,J_{7} & = 275\,{J_{2}}^{2}J_{3}, & 4802\,J_{8} &= 165\,{J_{2}}^{4}, & 12348\,J_{9} & = 3025\,{J_{2}}^{3}J_{3}, & 67228\,J_{10} & = 1815\,{J_{2}}^{5}.
\end{aligned}
\end{equation}
$D$ is in the transversely isotropic class $[\OO(2)]$ if moreover $J_{2}(D) \ne 0$. In that case, it admits the normal form~\eqref{eq:transversly_isotropic_case} where $\delta = {7\,J_{3}(D)}/{18\,J_{2}(D)}$.
\end{prop}

As for $[\octa]$ the expression of $\delta$ ensures that, for all $J_{2},J_{3} \in\RR$, the solution $\delta$ is real. Therefore no inequality is needed. This fact is indeed general to any linear slice of degree $1$.

\begin{rem}
All 2nd-order tensors in $\HH^{2} $ covariant to $D$ are multiple of $\diag(1,1,-2)$ and are eigenvectors of $D$
corresponding to the eigenvalue $12 \delta$.
\end{rem}

% ----------------------------------------------------------------

\subsection{Trigonal symmetry ($[\DD_{3}]$)}
\label{subsec:trigonal}

A $\DD_{3}$-invariant tensor $D \in \HH^{4} $ has the following matrix representation:
\begin{equation}\label{eq:trigonal_case}
  \underline{D} = \left(
   \begin {array}{cccccc}
    3\,\delta & \delta & -4\,\delta & -\sqrt{2}\sigma & 0 & 0 \\
    \delta & 3\,\delta & -4\,\delta & \sqrt{2}\sigma & 0 & 0 \\
    -4\,\delta & -4\,\delta & 8\,\delta & 0 & 0 & 0 \\
    -\sqrt{2}\sigma & \sqrt{2}\sigma & 0 & -8\,\delta & 0 & 0 \\
    0 & 0 & 0 & 0 & -8\,\delta & -2\,\sigma \\
    0 & 0 & 0 & 0 & -2\,\sigma & 2\,\delta
   \end {array}
  \right)
\end{equation}
where $(\delta, \sigma) \in \RR^{2}$. The action of the monodromy group $\GS_{2}$ is given by:
\begin{equation*}
    \delta \mapsto \delta, \qquad \sigma \mapsto -\sigma \, .
\end{equation*}
The evaluation of the invariants~\eqref{eq:boehler_invariants} on this slice gives:
\begin{align*}
  J_{2} & = 280\,{\delta}^{2}+16\,{\sigma}^{2} \\
  J_{3} & = 144\,\delta\, \left( 5\,{\delta}^{2}-{\sigma}^{2} \right) \\
  J_{4} & = 32800\,{\delta}^{4}+2720\,{\sigma}^{2}{\delta}^{2}+88\,{\sigma}^{4} \\
  J_{5} & = 32\,\delta\, \left( -{\sigma}^{2}+50\,{\delta}^{2} \right)^{2} \\
  J_{6} & = 4528000\,{\delta}^{6}+436800\,{\sigma}^{2}{\delta}^{4}+20640\,{\sigma}^{4}{\delta}^{2}+496\,{\sigma}^{6} \\
  J_{7} & = 320\,\delta\, \left( 22\,{\delta}^{2}+{\sigma}^{2} \right) \left( -{\sigma}^{2}+50\,{\delta}^{2} \right)^{2} \\
  J_{8} & = 3840\,{\delta}^{2} \left( 22\,{\delta}^{2}+{\sigma}^{2} \right) \left( -{\sigma}^{2}+50\,{\delta}^{2} \right)^{2} \\
  J_{9} & = 3200\,\delta\, \left( -{\sigma}^{2}+50\,{\delta}^{2} \right)^{2} \left( 22\,{\delta}^{2}+{\sigma}^{2} \right)^{2}  \\
  J_{10} & = 38400\,{\delta}^{2} \left( -{\sigma}^{2}+50\,{\delta}^{2} \right) ^{2} \left( 22\,{\delta}^{2}+{\sigma}^{2} \right)^{2}
\end{align*}
We notice that the parameter:
\begin{equation*}
    \delta=-\frac{1}{4}\,{\frac {J_{5}}{{J_{2}}^{2}-3\,J_{4}}}
\end{equation*}
is a rational invariant, therefore for all $J_{2},J_{4},J_{5}\in\RR$, the parameter $\delta$ is real. However, the minimal equation satisfied by $\sigma$ is of degree 2
\begin{equation}\label{eq:Syz_D3}
    J_{2}=280\,{\delta}^{2}+16\,{\sigma}^{2}\,.
\end{equation}

A condition is needed for $\sigma$ to be real. This condition, which is $\sigma^{2}\ge 0$, is given by the following inequality on $J_{2},J_{4},J_{5}$:
\begin{equation}\label{eq:Ineq_D3}
    2J_{2}(J_{2}^{2} - 3\,J_{4})^{2}-35J_{5}^{2} \ge 0
\end{equation}

Having computed $\delta$, the equation \eqref{eq:Syz_D3} has two roots with opposite signs. The linear slice corresponding to the group $\DD_{3}$ is of degree 2, which was already known (see \autoref{sec:normalizers}).

\begin{prop}\label{prop:trigonal}
A harmonic tensor $D \in \HH^{4}$ is in the closed stratum $\cstrata{\DD_{3}}$ if and only if the invariants $J_{2}(D) \cdots J_{10}(D)$ satisfy the following polynomial relations
\begin{equation}\label{eq:trigonal_syzigies}
\begin{aligned}
192\,J_{{6}}  &=  -51\,{J_{{2}}}^{3}+216\,J_{{2}}J_{{4}}+10\,{J_{{3}}}^{2}
\\
36\,J_{{7}}  &=  -2\,{J_{{2}}}^{2}J_{{3}}+6\,J_{{3}}J_{{4}}+27\,J_{{2}}J_{{
5}}
\\
768\,{J_{{4}}}^{2}  &=  -99\,{J_{{2}}}^{4}+552\,{J_{{2}}}^{2}J_{{4}}+10\,J_
{{2}}{J_{{3}}}^{2}+240\,J_{{3}}J_{{5}}
\\
240\,J_{{8}}  &=  -33\,{J_{{2}}}^{4}+96\,{J_{{2}}}^{2}J_{{4}}+30\,J_{{2}}{J
_{{3}}}^{2}+40\,J_{{3}}J_{{5}}
\\
576\,J_{{4}}J_{{5}}  &=  -41\,{J_{{2}}}^{3}J_{{3}}+120\,J_{{2}}J_{{3}}J_{{4
}}+216\,{J_{{2}}}^{2}J_{{5}}+30\,{J_{{3}}}^{3}
\\
1152\,J_{{9}}  &=  -99\,{J_{{2}}}^{3}J_{{3}}+296\,J_{{2}}J_{{3}}J_{{4}}+648
\,{J_{{2}}}^{2}J_{{5}}+10\,{J_{{3}}}^{3}
\\
1440\,{J_{{5}}}^{2}  &=  -11\,{J_{{2}}}^{5}+32\,{J_{{2}}}^{3}J_{{4}}-70\,{J
_{{2}}}^{2}{J_{{3}}}^{2}+240\,J_{{2}}J_{{3}}J_{{5}}+240\,{J_{{3}}}^{2}
J_{{4}}
\\
8640\,J_{{10}}  &=  -891\,{J_{{2}}}^{5}+2592\,{J_{{2}}}^{3}J_{{4}}+730\,{J_
{{2}}}^{2}{J_{{3}}}^{2}+2160\,J_{{2}}J_{{3}}J_{{5}}+240\,{J_{{3}}}^{2}
J_{{4}}
\end{aligned}
\end{equation}
together with inequality~\eqref{eq:Ineq_D3}. $D$ is in the trigonal class $[\DD_{3}]$ if moreover $3\,J_{4} - J_{2}^{2} \ne 0$ and $98\,J_{4} - 41\,J_{2}^{2} \ne 0$. In that case, it admits the normal form~\eqref{eq:trigonal_case} where
\begin{equation*}
    \delta=-\frac{1}{4}\,{\frac {J_{5}}{{J_{2}}^{2} - 3\,J_{4}}}
\end{equation*}
and where $\sigma$ is the positive root of $J_{2}=280\,{\delta}^{2}+16\,{\sigma}^{2}$.
\end{prop}

\begin{rem}
All 2nd-order tensors in $\HH^{2}$ covariant to $D$ are multiple of $\diag(1,1,-2)$ and are eigenvectors of $D$
corresponding to the eigenvalue $12 \delta$.
\end{rem}

% ----------------------------------------------------------------

\subsection{Tetragonal symmetry ($[\DD_{4}]$)}
\label{subsec:tetragonal}

A $\DD_{4}$-invariant tensor $D \in \HH^{4}$ has the following matrix representation:
\begin{equation}\label{eq:tetragonal_case}
  \underline{D} = \left(
   \begin {array}{cccccc}
    -\sigma+3\,\delta & \sigma+\delta & -4\,\delta & 0 & 0 & 0 \\
    \sigma+\delta & -\sigma+3\,\delta & -4\,\delta & 0 & 0 & 0 \\
    -4\,\delta & -4\,\delta & 8\,\delta & 0 & 0 & 0 \\
    0 & 0 & 0 & -8\,\delta & 0 & 0 \\
    0 & 0 & 0 & 0 & -8\,\delta & 0 \\
    0 & 0 & 0 & 0 & 0 & 2\,\sigma+2\,\delta
   \end {array}
  \right)
\end{equation}
where $(\delta, \sigma) \in \RR^{2}$. The monodromy group is defined by:
\begin{equation*}
    \delta \mapsto \delta, \qquad \sigma \mapsto -\sigma .
\end{equation*}

The evaluation of the invariants~\eqref{eq:boehler_invariants} on this slice gives:
\begin{align*}
    J_{2} & = 8\,{\sigma}^{2}+280\,{\delta}^{2} \\
    J_{3} & = 48\,\delta\, \left( {\sigma}^{2}+15\,{\delta}^{2} \right) \\
    J_{4} & = 32\,{\sigma}^{4}+960\,{\delta}^{2}{\sigma}^{2}+32800\,{\delta}^{4} \\
    J_{5} & = 128\,\delta\, \left( 5\,\delta-\sigma \right)^{2} \left( 5\,\delta+\sigma \right)^{2} \\
    J_{6} & = 128\,{\sigma}^{6}+5760\,{\sigma}^{4}{\delta}^{2}+86400\,{\sigma}^{2}{\delta}^{4}+4528000\,{\delta}^{6} \\
    J_{7} & = 512\,\delta\, \left( 55\,{\delta}^{2}+{\sigma}^{2} \right)\left( 5\,\delta-\sigma \right)^{2} \left( 5\,\delta+\sigma \right)^{2} \\
    J_{8} & = 6144\,{\delta}^{2} \left( 55\,{\delta}^{2}+{\sigma}^{2}\right) \left( 5\,\delta-\sigma \right)^{2} \left( 5\,\delta+\sigma \right)^{2} \\
    J_{9} & = 2048\,\delta\, \left( 5\,\delta-\sigma \right) ^{2} \left( 5\,\delta+\sigma \right)^{2} \left( 55\,{\delta}^{2}+{\sigma}^{2}\right) ^{2} \\
    J_{10} & = 24576\,{\delta}^{2} \left( 5\,\delta-\sigma \right) ^{2}\left( 5\,\delta+\sigma \right)^{2} \left( 55\,{\delta}^{2}+{\sigma}^{2} \right)^{2}
\end{align*}
We notice that the parameter
\begin{equation*}
    \delta=-\frac{1}{4}\,{\frac {J_{5}}{{J_{2}}^{2} - 3\,J_{4}}}
\end{equation*}
is a rational invariant, therefore for all $(J_{2},J_{4},J_{5})\in\RR^{3}$, $\delta\in\RR$. However, the minimal equation satisfied by $\sigma$ is of degree 2
\begin{equation}\label{eq:Syz_D4}
    J_{2}=8\,{\sigma}^{2}+280\,{\delta}^{2}
\end{equation}
As for $[\DD_{3}]$, the condition
\begin{equation}\label{eq:Ineq_D4}
2J_{2}(J_{2}^{2} - 3\,J_{4})^{2}-35J_{5}^{2}\ge0
\end{equation}
has to be added in order to ensure $\sigma$ to be real.

Having computed $\delta$, the equation \eqref{eq:Syz_D4} has two roots with opposite signs. The linear slice corresponding to the group $\DD_{4}$ is of degree 2, which was already known (see \autoref{sec:normalizers}).

\begin{prop}\label{prop:tetragonal}
A harmonic tensor $D \in \HH^{4} $ is in the closed stratum $\cstrata{\DD_{4}}$ if and only if the invariants $J_{2}(D) \cdots J_{10}(D)$ satisfy the following polynomial relations
\begin{equation}\label{eq:tetragonal_syzigies}
\begin{aligned}
6\,J_{{6}}  &=  -3\,{J_{{2}}}^{3}+9\,J_{{2}}J_{{4}}+20\,{J_{{3}}}^{2}
\\
3\,J_{{7}}  &=  {J_{{2}}}^{2}J_{{3}}-3\,J_{{3}}J_{{4}}+3\,J_{{2}}J_{{5}}
\\
6\,{J_{{4}}}^{2}  &=  -3\,{J_{{2}}}^{4}+9\,{J_{{2}}}^{2}J_{{4}}+20\,J_{{2}}
{J_{{3}}}^{2}-20\,J_{{3}}J_{{5}}
\\
5\,J_{{8}}  &=  -3\,{J_{{2}}}^{4}+6\,{J_{{2}}}^{2}J_{{4}}+30\,J_{{2}}{J_{{3
}}}^{2}-5\,J_{{3}}J_{{5}}
\\
3\,J_{{4}}J_{{5}}  &=  7\,{J_{{2}}}^{3}J_{{3}}-15\,J_{{2}}J_{{3}}J_{{4}}+3
\,{J_{{2}}}^{2}J_{{5}}-60\,{J_{{3}}}^{3}
\\
6\,J_{{9}}  &=  5\,{J_{{2}}}^{3}J_{{3}}-13\,J_{{2}}J_{{3}}J_{{4}}+6\,{J_{{2
}}}^{2}J_{{5}}-20\,{J_{{3}}}^{3}
\\
5\,{J_{{5}}}^{2}  &=  -2\,{J_{{2}}}^{5}+4\,{J_{{2}}}^{3}J_{{4}}+10\,{J_{{2}
}}^{2}{J_{{3}}}^{2}-20\,J_{{2}}J_{{3}}J_{{5}}+30\,{J_{{3}}}^{2}J_{{4}}
\\
15\,J_{{10}}  &=  -9\,{J_{{2}}}^{5}+18\,{J_{{2}}}^{3}J_{{4}}+85\,{J_{{2}}}^
{2}{J_{{3}}}^{2}-30\,J_{{2}}J_{{3}}J_{{5}}+15\,{J_{{3}}}^{2}J_{{4}}
\end{aligned}
\end{equation}
together with inequality~\eqref{eq:Ineq_D4}. $D$ is in the tetragonal class $[\DD_{4}]$ if moreover $3\,J_{4} - J_{2}^{2} \ne 0$ and $5\,J_{2}^{3} - 8\,J_{2}J_{4} - 70\,J_{3}^{2} \ne 0$. In that case, it admits the normal form~\eqref{eq:tetragonal_case} where
\begin{equation*}
    \delta = -\frac{1}{4}\,{\frac {J_{5}}{{J_{2}}^{2} - 3\,J_{4}}}
\end{equation*}
and where $\sigma$ is the positive root of $J_{2} = 8\,{\sigma}^{2} + 280\,{\delta}^{2}$.
\end{prop}

\begin{rem}
All 2nd-order tensors in $\HH^{2} $ covariant to $D$ are multiple of $\diag(1,1,-2)$ and are eigenvectors of $D$
corresponding to the eigenvalue $12 \delta$.
\end{rem}

% ----------------------------------------------------------------

\subsection{Orthotropic symmetry ($[\DD_{2}]$)}
\label{subsec:orthotropic}

A $\DD_{2}$-invariant tensor $D \in \HH^{4} $ has the following matrix representation:
\begin{equation}\label{eq:orthotropic_case}
  \underline{D} = \left(
   \begin {array}{cccccc}
    -\lambda_{2} - \lambda_{3} &\lambda_{3} & \lambda_{2} & 0 & 0 & 0 \\
    \lambda_{3} & -\lambda_{3}-\lambda_{1} & \lambda_{1} & 0 & 0 & 0 \\
    \lambda_{2} & \lambda_{1} & -\lambda_{1}-\lambda_{2} & 0 & 0 & 0 \\
    0 & 0 & 0 & 2 \lambda_{1} & 0 & 0 \\
    0 & 0& 0 & 0 & 2\lambda_{2} & 0 \\
    0 & 0 & 0 & 0 & 0 & 2\lambda_{3}
   \end {array}
  \right)
\end{equation}
where $(\lambda_{1},\lambda_{2},\lambda_{3}) \in \RR^{3}$. The monodromy group is the group $\GS_{3}$ acting on $(\lambda_{1}, \lambda_{2}, \lambda_{3})$ by permutation.
The evaluations of the invariants~\eqref{eq:boehler_invariants} on this slice are $\GS_{3}$-invariant.
They can therefore be expressed as polynomial functions of the elementary symmetric polynomials ($\sigma_{1},\sigma_{2},\sigma_{3}$) defined by:
\begin{equation*}
    \sigma_{1} := \lambda_{1} + \lambda_{2} + \lambda_{3},
    \qquad
    \sigma_{2} := \lambda_{1} \lambda_{2} + \lambda_{1} \lambda_{3} + \lambda_{2} \lambda_{3},
    \qquad
    \sigma_{3} := \lambda_{1} \lambda_{2} \lambda_{3}.
\end{equation*}
More precisely, we have:
\begin{equation*}
\begin{aligned}
    J_{2} & = - 14\,\sigma_{2} + 8\,\sigma_{1}^{2} \\
    J_{3} & = - 6\,\sigma_{1}\sigma_{2} + 24\,\sigma_{3}\\
    J_{4} & = 40\,\sigma_{1}\sigma_{3} - 112\,\sigma_{1}^{2}\sigma_{2} + 68\,\sigma_{2}^{2} + 32\,\sigma_{1}^{4}\\
    J_{5} & = 64\,\sigma_{1}^{2}\sigma_{3} - 12\,\sigma_{2}\sigma_{3} - 16\,\sigma_{1}^{3}\sigma_{2} + 28\,\sigma_{1}\sigma_{2}^{2} \\
    J_{6} & = - 344\,\sigma_{2}^{3} + 192\,\sigma_{1}^{3}\sigma_{3} - 24\,\sigma_{3}^{2} - 672\,\sigma_{1}^{4}\sigma_{2} + 1008\,\sigma_{1}^{2}\sigma_{2}^{2} + 128\,\sigma_{1}^{6} - 504\,\sigma_{1}\sigma_{2}\sigma_{3} \\
    J_{7} & = - 432\,\sigma_{1}^{2}\sigma_{2}\sigma_{3} + 384\,\sigma_{1}^{4}\sigma_{3} + 104\,\sigma_{2}^{2}\sigma_{3} - 96\,\sigma_{1}\sigma_{3}^{2} - 64\,\sigma_{1}^{5}\sigma_{2} + 192\,\sigma_{1}^{3}\sigma_{2}^{2} - 248\,\sigma_{1}\sigma_{2}^{3} \\
    J_{8} & = 608\,\sigma_{1}^{3}\sigma_{2}\sigma_{3} + 80\,\sigma_{2}^{4} - 768\,\sigma_{1}^{5}\sigma_{3} + 192\,\sigma_{1}^{2}\sigma_{3}^{2} + 72\,\sigma_{2}\sigma_{3}^{2} + 288\,\sigma_{1}^{4}\sigma_{2}^{2} - 416\,\sigma_{1}^{2}\sigma_{2}^{3} + 744\,\sigma_{1}\sigma_{2}^{2}\sigma_{3}\\
    J_{9} & =  - 5248\,\sigma_{1}^{4}\sigma_{2}\sigma_{3} + 2880\,\sigma_{1}^{2}\sigma_{2}^{2}\sigma_{3} + 1328\,\sigma_{1}\sigma_{2}\sigma_{3}^{2} + 144\,\sigma_{3}^{3} + 2304\,\sigma_{1}^{6}\sigma_{3} - 1152\,\sigma_{1}^{3}\sigma_{3}^{2} - 880\,\sigma_{2}^{3}\sigma_{3} \\
    & \quad - 256\,\sigma_{1}^{7}\sigma_{2} + 1024\,\sigma_{1}^{5}\sigma_{2}^{2} - 2304\,\sigma_{1}^{3}\sigma_{2}^{3} + 2160\,\sigma_{1}\sigma_{2}^{4}\\
    J_{10} & = 10752\,\sigma_{1}^{5}\sigma_{2}\sigma_{3} - 1280\,\sigma_{1}^{3}\sigma_{2}^{2}\sigma_{3} - 5664\,\sigma_{1}\sigma_{2}^{3}\sigma_{3} - 2688\,\sigma_{1}^{2}\sigma_{2}\sigma_{3}^{2} - 800\,\sigma_{2}^{5} - 4608\,\sigma_{1}^{7}\sigma_{3} \\
    & \quad + 2304\,\sigma_{1}^{4}\sigma_{3}^{2} - 1344\,\sigma_{2}^{2}\sigma_{3}^{2} - 288\,\sigma_{1}\sigma_{3}^{3} + 1536\,\sigma_{1}^{6}\sigma_{2}^{2} - 4224\,\sigma_{1}^{4}\sigma_{2}^{3} + 3104\,\sigma_{1}^{2}\sigma_{2}^{4}
\end{aligned}
\end{equation*}
As it can be observed, the parametric system obtained for $[\DD_{2}]$ is rather complicated. Nevertheless, the evaluation of $J_{2}, \dotsc , J_{10}$ on the slice and their expressions using $\sigma_{1}$, $\sigma_{2}$, $\sigma_{3}$ is simple using computer. By computing the elimination ideal of a Groebner basis for the ideal generated by
\begin{equation*}
    J_{k} - p_{k}(\sigma_{1},\sigma_{2},\sigma_{3}), \qquad k=2, \dots , 7,
\end{equation*}
we obtain a set of 6 syzygies.

\begin{subequations}
\label{eqs:Syz_D2}
\begin{multline}\label{eq:Syz_D2_a}
   - 1350\,J_{3}\,J_{7} - 840\,J_{4}\,J_{6} + 465\,{J_{2}}^{2}\,J_{6} + 270\,{J_{5}}^{2} + 720\,J_{2}\,J_{3}\,J_{5} + 747\,J_{2}\,{J_{4}}^{2} \\
   - 170\,{J_{3}}^{2}\,J_{4} - 564\,{J_{2}}^{3}\,J_{4} + 70\,{J_{2}}^{2}\,{J_{3}}^{2} + 84\,{J_{2}}^{5} = 0
\end{multline}
\begin{multline}\label{eq:Syz_D2_b}
   - 1620\,J_{4}\,J_{7} + 810\,{J_{2}}^{2}\,J_{7} + 360\,J_{5}\,J_{6} - 1110\,J_{2}\,J_{3}\,J_{6} + 999\,J_{2}\,J_{4}\,J_{5} + 960\,{J_{3}}^{2}\,J_{5}
   \\
   - 549\,{J_{2}}^{3}\,J_{5} - 972\,J_{3}\,{J_{4}}^{2} + 1638\,{J_{2}}^{2}\,J_{3}\,J_{4} - 80\,J_{2}\,{J_{3}}^{3} - 312\,{J_{2}}^{4}\,J_{3} = 0
\end{multline}
\begin{multline}\label{eq:Syz_D2_c}
   4050\,J_{5}\,J_{7} - 25650\,J_{2}\,J_{3}\,J_{7} - 14310\,J_{2}\,J_{4}\,J_{6} + 9600\,{J_{3}}^{2}\,J_{6} + 7965\,{J_{2}}^{3}\,J_{6}
   + 9450\,J_{3}\,J_{4}\,J_{5}
   \\
   + 10530\,{J_{2}}^{2}\,J_{3}\,J_{5} + 1134\,{J_{4}}^{3} + 11259\,{J_{2}}^{2}\,{J_{4}}^{2} - 12330\,J_{2}\,{J_{3}}^{2}\,J_{4} - 9018\,{J_{2}}^{4}\,J_{4}
   \\
   + 400\,{J_{3}}^{4} + 3270\,{J_{2}}^{3}\,{J_{3}}^{2} + 1350\,{J_{2}}^{6} = 0
\end{multline}
\begin{multline}\label{eq:Syz_D2_d}
   - 12150\,J_{2}\,J_{3}\,J_{7} + 3600\,{J_{6}}^{2} - 11610\,J_{2}\,J_{4}\,J_{6} + 9750\,{J_{3}}^{2}\,J_{6} + 4410\,{J_{2}}^{3}\,J_{6}
   \\
   + 8505\,J_{3}\,J_{4}\,J_{5} + 3645\,{J_{2}}^{2}\,J_{3}\,J_{5} + 1458\,{J_{4}}^{3} + 5670\,{J_{2}}^{2}\,{J_{4}}^{2}
   \\
   - 10710\,J_{2}\,{J_{3}}^{2}\,J_{4} - 4104\,{J_{2}}^{4}\,J_{4} + 400\,{J_{3}}^{4} + 2580\,{J_{2}}^{3}\,{J_{3}}^{2} + 576\,{J_{2}}^{6} = 0
\end{multline}
\begin{multline}\label{eq:Syz_D2_e}
   1800\,J_{6}\,J_{7} - 10800\,J_{2}\,J_{4}\,J_{7} + 4800\,{J_{3}}^{2}\,J_{7} + 4950\,{J_{2}}^{3}\,J_{7} + 4020\,J_{3}\,J_{4}\,J_{6} - 8370\,{J_{2}}^{2}\,J_{3}\,J_{6}
   \\
   + 162\,{J_{4}}^{2}\,J_{5} + 7371\,{J_{2}}^{2}\,J_{4}\,J_{5} + 2880\,J_{2}\,{J_{3}}^{2}\,J_{5} - 3483\,{J_{2}}^{4}\,J_{5} - 9216\,J_{2}\,J_{3}\,{J_{4}}^{2}
   \\
   + 640\,{J_{3}}^{3}\,J_{4} + 11946\,{J_{2}}^{3}\,J_{3}\,J_{4} - 720\,{J_{2}}^{2}\,{J_{3}}^{3} - 2160\,{J_{2}}^{5}\,J_{3} = 0
\end{multline}
\begin{multline}\label{eq:Syz_D2_f}
   60750\,{J_{7}}^{2} + 178200\,J_{3}\,J_{4}\,J_{7} - 546750\,{J_{2}}^{2}\,J_{3}\,J_{7} + 3780\,{J_{4}}^{2}\,J_{6} - 246780\,{J_{2}}^{2}\,J_{4}\,J_{6}
   \\
    + 348000\,J_{2}\,{J_{3}}^{2}\,J_{6} + 137025\,{J_{2}}^{4}\,J_{6} + 116640\,J_{2}\,J_{3}\,J_{4}\,J_{5} - 75600\,{J_{3}}^{3}\,J_{5}
   \\
    + 223560\,{J_{2}}^{3}\,J_{3}\,J_{5} + 29808\,J_{2}\,{J_{4}}^{3} + 82170\,{J_{3}}^{2}\,{J_{4}}^{2} + 177660\,{J_{2}}^{3}\,{J_{4}}^{2}
   \\
   - 438390\,{J_{2}}^{2}\,{J_{3}}^{2}\,J_{4} - 148014\,{J_{2}}^{5}\,J_{4} + 17200\,J_{2}\,{J_{3}}^{4} + 102000\,{J_{2}}^{4}\,{J_{3}}^{2} + 22221\,{J_{2}}^{7} = 0
\end{multline}
\end{subequations}

If these relations are satisfied, the elementary symmetric polynomials $\sigma_{1}$, $\sigma_{2}$ and $\sigma_{3}$ can be expressed, on the open stratum defined by
\begin{equation*}
    6\,J_{6}-9\,J_{2}\,J_{4}-20\,{J_{3}}^{2}+3\,{J_{2}}^{3} \ne 0 ,
\end{equation*}
as rational expressions of $J_{2},\dotsc,J_{7}$. More precisely, we have
\begin{equation*}
	\sigma_{{2}}= \frac{4}{7}\,{\sigma_{{1}}}^{2}-\frac{1}{14}\,J_{{2}}, \quad
	\sigma_{{3}}= \frac{1}{24}\,J_{{3}} + \frac{1}{7}\,{\sigma_{{1}}}^{3}-{\frac {1}{56}}\,\sigma_{{1}}\,J_{{2}},
\end{equation*}
and
\begin{equation*}
    \sigma_{1} = - \frac{9\,\left( 3\,J_{7}-3\,J_{2}\,J_{5}+3\,J_{3}\,J_{4}-{J_{2}}^{2}\,J_{3}\right) }{2\,\left( 6\,J_{6}-9\,J_{2}\,J_{4}-20\,{J_{3}}^{2}+3\,{J_{2}}^{3}\right) }.
\end{equation*}

Then $\lambda_{1},\lambda_{2},\lambda_{3}$ are recovered (up to their monodromy group) as the roots of the third degree polynomial in $\lambda$
\begin{equation*}
    p(\lambda)=\lambda^{3}-\sigma_{1}\lambda^{2}+\sigma_{2}\lambda-\sigma_{3}
\end{equation*}
But to ensure that given $J_{2},\dotsc ,J_{7}\in\RR$, we can find $\lambda_{1},\lambda_{2},\lambda_{3}\in\RR$, the following inequalities have to be satisfied
\begin{equation*}
    \Delta_{3} = \sigma_{1}^{2}\sigma_{2}^{2}-4\sigma_{1}^{3}\sigma_{3}+18\sigma_{1}\sigma_{2}\sigma_{3}^{2}-27\sigma_{3}^{2} \ge 0 \, ,
\end{equation*}
as well as
\begin{equation*}
    \Delta_{2} = 2\sigma_{1}^{2} - 6\sigma_{2} \ge 0 \, .
\end{equation*}
The discriminant $\Delta_{3}$ can be rewritten, using $J_{2},\dotsc,J_{6}$, as
\begin{equation*}
    \Delta_{3} = \frac{6\,J_{6}-9\,J_{2}\,J_{4}-20\,{J_{3}}^{2}+3\,{J_{2}}^{3}}{432} \, .
\end{equation*}
whereas $\Delta_{2}$ can be rewritten, using $J_{2},\dotsc,J_{7}$, as
\begin{equation*}
    \Delta_{2} = \frac{N_{2}}{14\,{ \left( 6\,J_{6} - 9\,J_{2}\,J_{4} - 20\,{J_{3}}^{2} + 3\,{J_{2}}^{3}\right) }^{2}}
\end{equation*}
where
\begin{equation*}
\begin{split}
N_{2} & = 3645\,{J_{7}}^{2} + 7290\,J_{2}\,J_{5}\,J_{7} - 7290\,J_{3}\,J_{4}\,J_{7} + 2430\,{J_{2}}^{2}\,J_{3}\,J_{7} - 36\,J_{2}\,{J_{6}}^{2}\\
      & \quad  + 108\,{J_{2}}^{2}\,J_{4}\,J_{6} + 240\,J_{2}\,{J_{3}}^{2}\,J_{6} - 36\,{J_{2}}^{4}\,J_{6} - 3645\,{J_{2}}^{2}\,{J_{5}}^{2} + 7290\,J_{2}\,J_{3}\,J_{4}\,J_{5} \\
      & \quad - 2430\,{J_{2}}^{3}\,J_{3}\,J_{5} - 3645\,{J_{3}}^{2}\,{J_{4}}^{2} - 81\,{J_{2}}^{3}\,{J_{4}}^{2} + 2070\,{J_{2}}^{2}\,{J_{3}}^{2}\,J_{4} \\
      & \quad + 54\,{J_{2}}^{5}\,J_{4} - 400\,J_{2}\,{J_{3}}^{4} - 285\,{J_{2}}^{4}\,{J_{3}}^{2} - 9\,{J_{2}}^{7}.
\end{split}
\end{equation*}

We will summarize these results in the following proposition.

\begin{prop}\label{prop:orthotropic}
A harmonic tensor $D \in \HH^{4} $ is in the closed stratum $\cstrata{\DD_{2}}$ if and only if the invariants $J_{2}(D), \cdots ,J_{10}(D)$ satisfy the polynomial relations~\eqref{eqs:Syz_D2}, completed with the 3 equations
\begin{equation}\label{eq:Syz_D2_g}
    S_{k}(J_{2}, \dotsc J_{10}) =0, \qquad k = 8, 9, 10
\end{equation}
obtained from $J_{k} - p_{k}(\sigma_{1},\sigma_{2},\sigma_{3})$, where we have substituted the $\sigma_{i}$ by their corresponding rational expression, and we have cleared the denominators. The inequalities
\begin{equation}\label{eq:discriminant}
    6\,J_{6}-9\,J_{2}\,J_{4}-20\,{J_{3}}^{2}+3\,{J_{2}}^{3} \ge 0
\end{equation}
as well as
\begin{equation}
\begin{split}
 & 3645\,{J_{7}}^{2} + 7290\,J_{2}\,J_{5}\,J_{7} - 7290\,J_{3}\,J_{4}\,J_{7} + 2430\,{J_{2}}^{2}\,J_{3}\,J_{7} - 36\,J_{2}\,{J_{6}}^{2} + 108\,{J_{2}}^{2}\,J_{4}\,J_{6} \\
 & + 240\,J_{2}\,{J_{3}}^{2}\,J_{6} - 36\,{J_{2}}^{4}\,J_{6} - 3645\,{J_{2}}^{2}\,{J_{5}}^{2} + 7290\,J_{2}\,J_{3}\,J_{4}\,J_{5} - 2430\,{J_{2}}^{3}\,J_{3}\,J_{5} \\
 & - 3645\,{J_{3}}^{2}\,{J_{4}}^{2} - 81\,{J_{2}}^{3}\,{J_{4}}^{2} + 2070\,{J_{2}}^{2}\,{J_{3}}^{2}\,J_{4} + 54\,{J_{2}}^{5}\,J_{4} \\
 &  - 400\,J_{2}\,{J_{3}}^{4} - 285\,{J_{2}}^{4}\,{J_{3}}^{2} - 9\,{J_{2}}^{7} \ge 0
\end{split}
\end{equation}
are also required. $D$ is in the orthotropic class $[\DD_{2}]$ if moreover inequality~\eqref{eq:discriminant} is strict. In that case, it admits the normal form~\eqref{eq:orthotropic_case} where $\lambda_{1},\lambda_{2},\lambda_{3}$ are the roots of the third degree polynomial
\begin{equation*}
    p(\lambda)=\lambda^{3}-\sigma_{1}\lambda^{2}+\sigma_{2}\lambda-\sigma_{3} .
\end{equation*}
\end{prop}

\begin{rem}
Even if it is the core of the method, it worths to note that, instead of $\lambda_{1},\lambda_{2},\lambda_{3}$, the 3 numbers $\sigma_{1},\sigma_{2},\sigma_{3}$ are uniquely defined on the open stratum: they are invariants of $D$ and can be written as rational expressions of $(J_{2},\dotsc,J_{7})$. It should be noticed that it is also possible to express $\sigma_{1},\sigma_{2},\sigma_{3}$ as rational functions of \emph{only} $(J_{2},\dotsc,J_{5})$ on an \emph{open subset} of $\strata{\DD_{2}}$. We will, however, not give these formulas here.
\end{rem}

\begin{rem}
All second rank tensors in $\HH^{2}$ covariant to $D$ commute together and the commutators of two of any of them vanishes.
\end{rem}

% ----------------------------------------------------------------

\subsection{Bifurcation conditions for tensors in $\HH^{4}$}
\label{subsec:bifurcation_H4}

An important observation has to be made. Mathematical results \cite{Shi67} tell us that the invariant algebra of $\HH^{4}$ is generated by the 9 fundamental invariants $(J_{2},\dotsc,J_{10})$. In this set 6 invariants $(J_{2}, \dotsc , J_{7})$ are algebraically independent, meanwhile the others are linked to the formers by polynomial relations.
As discussed before, it is a well-known fact that this algebra separates the orbits. But as demonstrated exhaustively in this section, for any class with a finite monodromy group\footnote{i.e. for all symmetry classes, except the monoclinic and the triclinic classes}, the first 6 invariants $(J_{2}, \dotsc , J_{7})$ are necessary and sufficient to separate the orbits inside each class. This observation can be summed-up by the following  important theorem.

\begin{thm}
For the following isotropy classes: the \emph{cubic}, the \emph{transversely isotropic}, the \emph{trigonal} the \emph{tetragonal} (and the trivial \emph{isotropic}) classes in $\HH^{4}(\RR^{3})$, the first 4 invariants $(J_{2}, \dotsc , J_{5})$ separate the orbits inside the class. For the \emph{orthotropic} class, the first 6 invariants $(J_{2}, \dotsc , J_{7})$ separate the orbits inside the class.
\end{thm}

We will now summarized the previous results by giving bifurcation conditions on the $J_{k}$ to explicit how to ``travel'' from a given isotropy class to another. More precisely, so far, we have given \emph{necessary and sufficient} conditions to belong to a closed stratum $\cstrata{H}$ with finite monodromy group. As we have seen, $\cstrata{H}$ is characterized by a set of polynomial relations (syzygies) involving the invariants $J_{k}$. These relations define \emph{implicit equations} which characterize the closed stratum~\footnote{Notice that a rigorous and complete classification of these closed strata would require also some real inequalities that we have not completely detailed here.}.
The set of relations that characterizes a given (closed) stratum generates an \emph{ideal} (c.f. \autoref{subsec:Grobner}), $I_{[H]}$ in $\RR[J_{2},\dotsc,J_{10}]$, and we have
\begin{equation*}
    [H_{1}]\preceq[H_{2}] \quad \Rightarrow \quad I_{[H_{2}]} \subset I_{[H_{1}]}
\end{equation*}
The bifurcation conditions, detailed on figure~\ref{fig:bifurcation}, correspond to relations (generators) which need to be added to $I_{[H_{1}]}$ to obtain $I_{[H_{2}]}$.

\begin{figure}[h]
\begin{center}
    \includegraphics{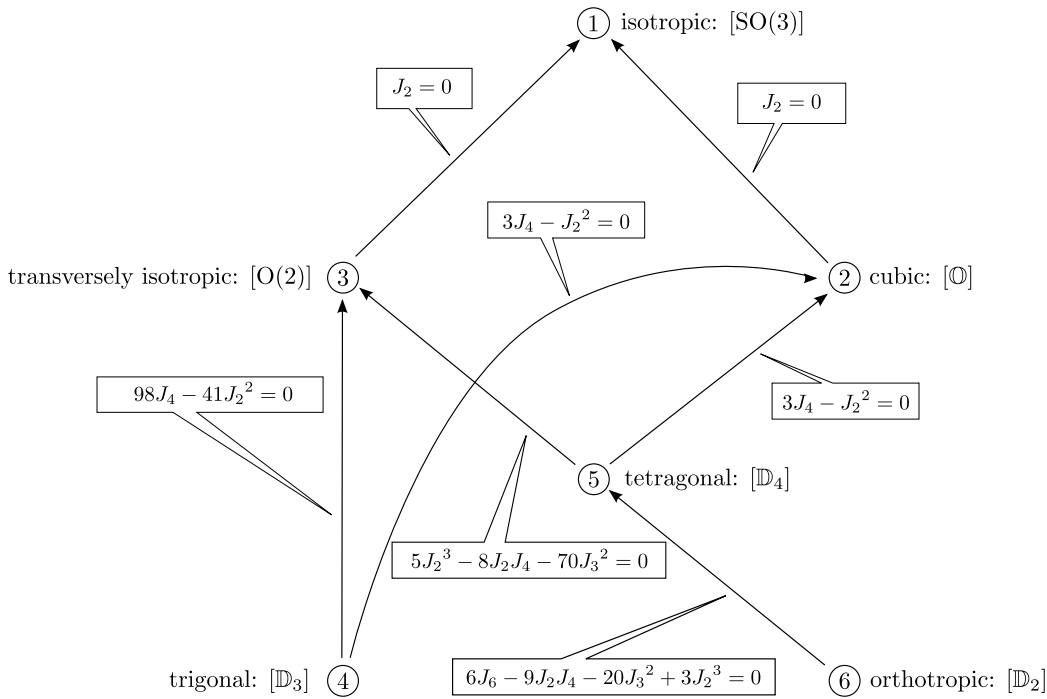}
    \caption{Bifurcation paths for finite monodromy isotropy classes in $\HH^{4}$.}
    \label{fig:bifurcation}
\end{center}
\end{figure}

\begin{rem}
These bifurcation relations were obtained with the following algorithm. Starting with some subgroup $H$ and the corresponding fixed point space $V^{H}$, we look for linear relations needed to define the subspace $V^{K}$ of $V^{H}$, where $H \subset K$. Given such a relation $\alpha = 0$ (in fact, in all the cases considered, only one relation was needed), we compute a $\Gamma^{H}$-invariant polynomial
\begin{equation*}
    p_{\alpha} := \prod_{\gamma\in \Gamma^{H}} \gamma \cdot \alpha \, ,
\end{equation*}
provided $\Gamma^{H}$ is finite. It is then clear that $p_{\alpha}$ vanishes on a point $v\in V^{H}$ if and only if $v$ belongs to $V^{K}$ or $V^{\gamma K \gamma^{-1}}$, for some $\gamma \in \Gamma^{H}$, that is if and only if $v$ has at least isotropy $[K]$. This $\Gamma^{H}$-invariant polynomial is then expressed using the $J_{k}$.
\end{rem}

% ----------------------------------------------------------------

\subsection{Relation with the characteristic polynomial}
\label{subsec:characteristic_polynomial}

To conclude this subsection, some links that exists between the invariants of $\HH^{4}$ and the characteristic polynomial of $\underline{D}$ will be investigated. The first link is formulated by the following proposition:
\begin{prop}\label{eq:characteristic_prop}
Let $\underline{D} \in \HH^{4}$.
The coefficients of the characteristic polynomial of $\underline{D}$ are {$\SO(3)$}-invariants of $D$.
They can be expressed in terms of the invariants $J_{2}, \dotsc , J_{5}$ defined above.
We have
\begin{equation}\label{eq:characteristic_pol}
    \chi_{D}(z) = {z}^{6} - \frac{1}{2}\,J_{2}{z}^{4} - \frac{1}{3}\,J_{3}{z}^{3} + \frac{1}{5}\,\left( {J_{2}}^{2} - 2\,J_{4} \right){z^{2}} + {\frac {2}{25}} \left( \,J_{2}J_{3} - {3}\,J_{5} \right)z .
\end{equation}
\end{prop}

\begin{proof}
Notice first that $z=0$ is always an eigenvalue of $\underline{D}$. Indeed, if $\mathbf{q}$ stands for the standard metric on $\RR^{3}$, we have $\underline{D}\mathbf{q} = \tr_{12}(D) = 0$. The coefficient of $z^{k}$ in $\chi_{D}$ is a homogeneous invariant polynomial of degree $6-k$ for $1\leq k \leq 4$. Therefore, it can be expressed as a polynomial  in $J_{2}(D),\cdots J_5(D)$. The computation of this last expression reduces therefore to identify some real coefficients. This can be done by specializing on particular values or on some particular strata where it is easy to compute, the cubic or the transverse isotropic stratum, for instance.
\end{proof}

Besides the characteristic polynomial, Betten \cite{Bet87} has introduced a two variables, invariant polynomial of the elasticity tensor. It is defined as
\begin{equation*}
    \mathcal{B}_{C}(\lambda, \mu) = \det( \underline{C} - \underline{C}_{\lambda,\mu})
\end{equation*}
where
\begin{equation*}
    \underline{C}_{\lambda,\mu} =
    \left(
      \begin{array}{cccccc}
        \lambda + 2\,\mu & \lambda & \lambda & 0 & 0 & 0 \\
        \lambda & \lambda + 2\,\mu & \lambda & 0 & 0 & 0 \\
        \lambda & \lambda & \lambda + 2\,\mu & 0 & 0 & 0 \\
        0 & 0 & 0 & 2\,\mu & 0 & 0 \\
        0 & 0 & 0 & 0 & 2\,\mu & 0 \\
        0 & 0 & 0 & 0 & 0 & 2\,\mu \\
      \end{array}
    \right)
\end{equation*}
is the totally isotropic tensor. Notice that $\mathcal{B}_{D}(0, \mu) = \chi_{D}(2\,\mu)$ and that $\mathcal{B}_{C}$ is of degree lower than 1 in $\lambda$ (indeed, this can be observed by subtracting the first line to the second line and to the third line in the determinant). The supposed interest of this polynomial is to obtain more invariants than the ones of the characteristic polynomial. However, it seems that it does not bring any new invariants for a tensor $D\in\HH^{4}(\RR^{3})$. Indeed, we have the following result.

\begin{lem}\label{lem:betten}
The Betten polynomial of a harmonic tensor $D\in \HH^{4}(\RR^{3})$ is related to its characteristic polynomial by the formula
\begin{equation*}
    \mathcal{B}_{D}(\lambda, \mu) = (3\,\lambda + 2\, \mu)\,\chi_{D}^{r}(2\,\mu)
\end{equation*}
where $\chi_{D}(\lambda) = \lambda\, \chi_{D}^{r}(\lambda)$.
\end{lem}

\begin{proof}
Let $L_{k}$ ($1\le k \le 6$) denote the lines of the determinant $\det( \underline{D} - \underline{C}_{\lambda,\mu})$. For a harmonic tensor $D$, we have
\begin{equation*}
    \sum_{k=1}^{3} d_{p,k} = 0, \qquad 1 \le p \le 6 .
\end{equation*}
Thus, substituting $L_{1} \to L_{1} + L_{2} + L_{3}$ in the determinant $\det( \underline{D} - \underline{C}_{\lambda,\mu})$, we get
\begin{equation*}
    \mathcal{B}_{D}(\lambda, \mu) = (3\,\lambda + 2\, \mu)\,P(\lambda, \mu)
\end{equation*}
but since $\mathcal{B}_{D}(\lambda, \mu)$ is of degree lower than 1 in $\lambda$ we have $P(\lambda, \mu) = P(0, \mu)$. Now for $\lambda = 0$ we have
\begin{equation*}
    \mathcal{B}_{D}(0, \mu) = \chi_{D}(2\,\mu) = (2\, \mu)\,P(0, \mu)
\end{equation*}
from which we get $P(0, \mu) = \chi_{D}^{r}(2\,\mu)$. This achieves the proof.
\end{proof}

\begin{cor}
The invariants defined by the coefficients of the Betten polynomial do not separate the orbits of $\Ela$.
\end{cor}

\begin{proof}
If it was the case, it would also separate the orbits of $\HH^{4}$ but then lemma~\ref{lem:betten} and proposition~\ref{eq:characteristic_prop} would imply that $J_{2}, J_{3}, J_{4}, J_{5}$ separate the orbits of $\HH^{4}$, which is false.
\end{proof}

% ----------------------------------------------------------------
% ----------------------------------------------------------------

\section{Towards an effective determinacy of elasticity classes}
\label{sec:stratifaction_ela}

In this last section, an application of the concepts previously introduced will be given. This application will concern the identification of the symmetry class of a tensor which is not expressed in one of its natural bases. As our objective is mainly to illustrate the tools we propose, we will only provide a set of sufficient conditions. This approach is based on a hypothesis of \emph{genericity} for the anisotropic tensors. Such hypothesis, which amounts to exclude any exceptional case, will be detailed and criticized in the next subsection. Therefore, in order to avoid any misunderstanding, our aim is not to provide a complete answer to identification problem. A full treatment of this question would require a detailed analysis of each particular non generic case that may occur, and constitutes a full and complex study on its own. As the core of this first paper is focused on the symmetry class identification in terms of $\HH^{4}$ invariants, we have rather choose a simpler problem to treat.

% ----------------------------------------------------------------

\subsection{Generic anisotropic elasticity tensors}
\label{subsec:generecity}

A well-known result of anisotropic elasticity~\cite{FV96} states that the symmetry group of an elasticity tensor is the intersection of the symmetry group of each one of its irreducible components. In  other terms:
\begin{equation*}
G_{\mathbf{C}}=G_{a}\cap G_{b}\cap G_{D}
\end{equation*}
since we trivially have $G_{\lambda}=G_{\mu}=\SO(3)$.
In the most general situation, all the former tensors are independent and therefore a complete analysis of the intersection of symmetry group has to be conducted. Such a problem is a complex one, and is a subject of a study on its own. For the physical situations considered here, a simplifying hypothesis will be made. As the tensor symmetry is a consequence of the material symmetries, a weak coupling between the tensors of the set $(a,b,D)$ will be considered.
To that end let us consider the triplet of quadratic forms $\mathbf{L}:=(a,b,\mathbf{d}_{2})$, where $\mathbf{d}_{2}$ is the fundamental covariant of degree $2$ introduced in~\eqref{eq:boehler_invariants}.
We have the following direct lemma
\begin{lem}\label{lem:LD}
\begin{equation*}
G_{\mathbf{C}}=G_{\mathbf{L}}\cap G_{D}
\end{equation*}
\end{lem}
\begin{proof}
As  $\mathbf{d}_{2}$ is covariant to  $D$, $G_{D}\subseteq G_{\mathbf{d}_{2}}$ and therefore $G_{D}\cap G_{\mathbf{d}_{2}}=G_{D}$.
Thus
\begin{equation*}
G_{\mathbf{L}}\cap G_{D} =  G_{a}\cap G_{b}\cap G_{\mathbf{d}_{2}}\cap G_{D} = G_{a}\cap G_{b}\cap G_{D}=G_{\mathbf{C}}
\end{equation*}
\end{proof}
The hypothesis that will be made here is the following
\begin{hyp}[Material Genericity Assumption (MGA)]
 If $[G_{\mathbf{L}}]\succeq[D_{2}]$ then $G_{\mathbf{L}}=G_{\mathbf{d}_{2}}$
\end{hyp}
Such an hypothesis is equivalent to the following:
\begin{equation*}
    \text{if } [G_{\mathbf{L}}]\succeq[D_{2}] \text{ then } G_{\mathbf{d}_{2}}\subseteq G_{a} \text{ and } G_{\mathbf{d}_{2}}\subseteq G_{b}.
\end{equation*}
From a physical point of view, this means that the symmetries of $a$ and $b$ are compatible with the ones of $D$ in the sense that adding $G_{a}$ and $G_{b}$ to the intersection will not lower $G_{D}$. In other words :
\begin{lem}
\begin{equation*}
G_{\mathbf{L}}=G_{\mathbf{d}_{2}}\Rightarrow G_{C}= G_{D}
\end{equation*}
\end{lem}
\begin{proof}
From the lemma \ref{lem:LD}, we have $G_{\mathbf{C}}=G_{\mathbf{L}}\cap G_{D}$. Therefore, according to the hypothesis $G_{\mathbf{C}}=G_{\mathbf{d}_{2}}\cap G_{D}$, and as $\mathbf{d}_{2}$ is covariant to $D$ the intersection reduces to $G_{\mathbf{C}}=G_{D}$.
\end{proof}
If this genericity assumption fails, and without a deeper study, all it can be said is $G_{\mathbf{C}}\subseteq G_{\mathbf{D}}$.
From a mathematical point of view, studying these cases are necessary to close the problem.
On the other hand, from a material point of view and for a first approach, we believe that the genericity assumption to hold true for a broad class of classical materials.
Anyway situations where this hypothesis may fail can be conceived:
\begin{description}
\item[Tailored materials] As the symmetry class of an elasticity tensor is driven by the architecture of the microstructure, it can be conceived to tailor the microstructure in order to give the material some non-conventional properties \cite{He04}. Therefore for such a specific design the genericity can not be assumed. Such kind of very specific architecture might not only  be found in engineered materials, but also in some biomaterials.
\item[Damaged materials] Genericity may also fail for damaged-materials, according to certain model \cite{Fra11} the symmetry breaking induces by the damage process may result in non generic anisotropic properties.
\end{description}
Therefore, bearing in mind the aforementioned limitations induced by the genericity assumption, we can propose, in the following section, an identification algorithm.

% ----------------------------------------------------------------

\subsection{Symmetry class identification}
\label{subsec:identification}

Based on the results obtained in this paper, the algorithm provided figure \ref{fig:organigram} allows to identify the symmetry class under the MGA assumption.
\begin{figure}
  \centering
  \includegraphics{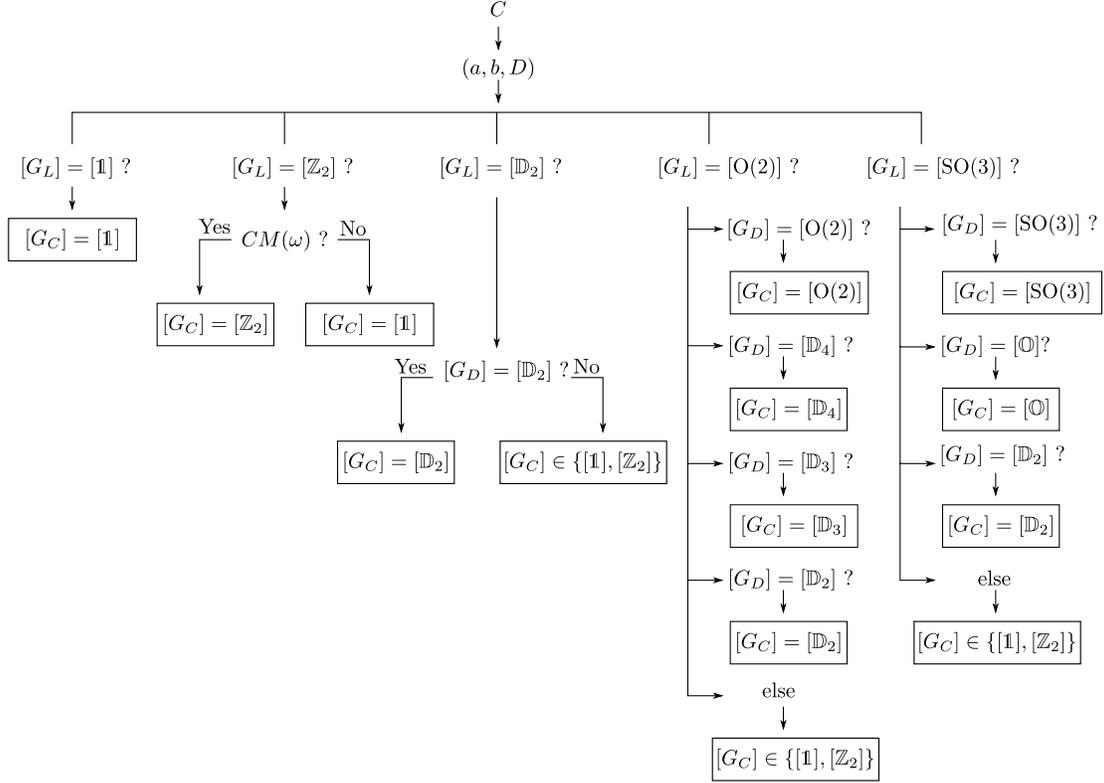}\\
  \caption{Symmetry identification for a generic tensor in $\Ela$.}
  \label{fig:organigram}
\end{figure}
Concerning this algorithm some remarks are in order.
\begin{itemize}
  \item The first step of the algorithm is based on the symmetry analysis of the triplet of quadratic forms $\mathbf{L}$. This symmetry identification is based on the conditions given in \autoref{subsec:n_quadratic_form}. As indicated on the diagram if $G_{\mathbf{L}}=\ZZ_{2}$, the symmetry class of $\mathbf{C}$ is decided using the Cowin-Mehrabadi criterion for the non vanishing vector $\omega$, obtained from case (1) of Theorem~\ref{thm:nS2_bifurcations} applied to the triplet $\mathbf{L}$. This test is denoted by $CM(\omega)?$ in figure~\ref{fig:organigram}.
  \item For classes strictly greater than $[\DD_{2}]$ we use the syzygy relations given in \autoref{sec:stratification_H4}, to check the isotropy class of $D$.
\end{itemize}

Some tests can be avoided \textit{a priori} (and will thus not be included in figure~\ref{fig:organigram}) by the following observations which result from lemma~\ref{lem:scalar_d2} combined with the bifurcation conditions given in figure~\ref{fig:bifurcation}.

\begin{enumerate}
  \item If $G_{L} = \OO(2) = G_{d_{2}}$, then the isotropy class of $D$ cannot belong to $\set{[\octa],[\SO(3)]}$.
  \item If $G_{L} = \SO(3) = G_{d_{2}}$, then $J_2^2 - 3J_4 = 0$ and the isotropy class of $D$ cannot belong to $\set{[\DD_{3}], [\DD_{4}], [\OO(2)]}$.
\end{enumerate}

\begin{lem}\label{lem:scalar_d2}
If $\mathbf{d}_2 =\lambda\mathrm{Id}$ then $J_2^2 - 3J_4 = 0$.
\end{lem}

\begin{proof}
If $\mathbf{d}_2 $ is a scalar multiple of the identity, $J_2^2=9\lambda^2$ and $J_{4}=\tr(\mathbf{d}_{4})=\tr(\mathbf{d}^2_{2})=3\lambda^{2}$. The conclusion follows.
\end{proof}

% ----------------------------------------------------------------
% ----------------------------------------------------------------

\section{Conclusion}
\label{sec:conclusion}

In this paper a method aiming at identifying the symmetry class of an elasticity tensor has been provided. Conversely to other approaches based on the spectral decomposition of the elasticity tensor, this method is based on the specific relations between invariants in each symmetry class. A complete set of relations were given for the couple of quadratic forms appearing in the harmonic decomposition of $\Ela$, and $5$ sets of algebraic relations were provided to identify the following symmetry classes from $\HH^{4}$: orthotropic ($[\DD_{2}]$), trigonal ($[\DD_{3}]$), tetragonal ($[\DD_{4}]$), transverse isotropic ($[\SO(2)]$) and cubic ($[\octa$]). Completed by a \emph{genericity} assumption on the elasticity tensor under study this approach was shown to be able to identify the $8$ possible elasticity classes.

We consider this paper as a first step towards a more systematic uses of invariant-based methods in continuum mechanics. Using the geometrical framework introduced in this paper some extensions of the current method can be envisaged :
\begin{itemize}
\item The first point would be to get rid of the genericity assumption we made, such a point constitutes a full study on its own ;
\item We aim at extending the current sygyzy-characterization to group having continuous monodromy group ;
\item To propose a complete set of normal forms of elasticity tensor constructed using this invariants approach ;
\item To tackle the increasing degree of the algebraic relations for low symmetry group, it might be interesting to characterize classes in terms of covariants.
\end{itemize}

Such an approach was conducted here in the context of classical elasticity, but the framework used allows a direct extension of this method to other kinds of anisotropic tensorial behaviors. It can be suited for the study of anisotropic features of piezoelasticity \cite{GW02}, flexoelasticity \cite{LQH11} or strain-gradient elasticity \cite{LQAHB12}. Furthermore, and in a more experimental view, we aim at testing our approach on experimental situations. This would allow to compare the proposed approach with the more classical ones.

% ----------------------------------------------------------------
% ----------------------------------------------------------------

\appendix

% ----------------------------------------------------------------
% ----------------------------------------------------------------

\section{Geometric concepts}
\label{sec:geometric_concepts}

% ----------------------------------------------------------------

\subsection{Normal form}
\label{subsec:normal_form}

A \emph{normal form} is a map $\NF: V \to V$ such that
\begin{equation*}
    v_{1} \approx v_{2} \text{ if and only if } \NF(v_{1}) = \NF(v_{2}).
\end{equation*}
The existence of a \emph{computable normal form} allows to decide \emph{whether} two vectors $v_{1}$ and $v_{2}$ \emph{belong to the} same $G$-orbit.

\begin{rem}
In practice, one contents itself usually with a normal form defined up to a \emph{finite group}. For instance, one is satisfied with a diagonal matrix as ``the'' normal form for the action of the rotation group on the space of symmetric matrices. However this diagonal form is not unique. It is defined up to the permutation of the diagonal entries.
\end{rem}

% ----------------------------------------------------------------

\subsection{Covariants}
\label{subsec:covariants}

Given a second representation $(V', \rho')$ of the same group $G$, we say that a linear map $\varphi: V \to V'$ is \emph{equivariant} if the following diagram
\begin{equation*}
    \begin{CD}
      V @> \varphi >> V' \\
      @V\rho(g) VV  @VV\rho'(g)V  \\
      V @> \varphi >> V'
\end{CD}
\end{equation*}
commutes for each $g \in G$. The composition of two equivariant maps is equivariant.

Given $v\in V$ the element $\varphi(v)$ will be said to be \emph{covariant} to $v$. It is clear that every element $v \in V$ is covariant to itself (with $\varphi = \Id$). However, this relation is not symmetric. Indeed, if $v'$ is covariant to $v$, the converse is not necessarily true ($v$ is not necessarily covariant to $v'$) unless the linear map $\varphi: V \longrightarrow V'$ is bijective. In this case, we say that the representations $(V, \rho)$ and $(V', \rho')$ are \emph{isomorphic} or \emph{equivalent}.

When the representation $(V', \rho')$ is \emph{trivial}, i.e. when $\rho'(g) = \Id$ for all $g\in G$, the covariant element $\varphi(v)$ will be called an \emph{invariant}. We have then
\begin{equation*}
    v \approx v' \Rightarrow \varphi(v) = \varphi(v').
\end{equation*}

% ----------------------------------------------------------------

\subsection{Polynomial invariants}
\label{subsec:polynomial_invariants}

The linear action of the group $G$ on the $\RR$-vector space $V$ extends naturally to the vector space of polynomial functions defined on $V$ with values in $\RR$. This extension is given by
\begin{equation*}
    (g\cdot P)(v) := P(g^{-1}\cdot v)
\end{equation*}
for every polynomial $P \in \RR[V]$ and every vector $v \in V$. The set of all \emph{invariant polynomials} is a sub-algebra of the algebra $\RR[V]$ of all polynomials defined on $V$. We denote this sub-algebra by $\RR[V]^{G}$ and call it the \emph{invariant algebra}.

When $G$ is a \emph{compact} Lie group (or more generally a \emph{reductive} group), this algebra $\RR[V]^{G}$ is of \emph{finite type} \cite{Hil93}. This means first that it is possible to find a finite set of invariant polynomials $J_{1}, \dotsc , J_{N}$ which generate (as an algebra) $\RR[V]^{G}$.
But, in general, this algebra is not \emph{free}, which means that it is not always possible to extract from $J_{1}, \dotsc , J_{N}$ a subset of generators which are \emph{algebraically independent}\footnote{Polynomials $P_{1}, \dotsc , P_{N}$ are said to be algebraically independent over $\RR$ if the only polynomial in $N$ variables which satisfies $Q(P_{1}, \dotsc , P_{N}) = 0$ is the zero polynomial.}.
The set of polynomial relations on $J_{1}, \dotsc , J_{N}$ (also known as \emph{syzygies}) is an ideal of $\RR[J_{1}, \dotsc ,J_{N}]$ that is finitely generated.
Such a set of generators is called an \emph{integrity basis} for $\RR[V]^{G}$.
Of course it is not unique and moreover its cardinal may change from one basis to another.

% ----------------------------------------------------------------

\subsection{Ideal and Groebner basis}
\label{subsec:Grobner}

Consider a subset $V$ of $k^{n}$ (where $k = \RR$ or $\CC$) defined implicitly by a finite set of polynomial equations $p_{1}, \dotsc p_{r}$. The fundamental observation is that the set $I$ of polynomials $p \in k[x_{1}, \dotsc , x_{n}]$ such that $p(x) = 0$ for each point $x \in S$ (where $x = (x_{1}, \dotsc , x_{n})$) satisfies the following properties:
\begin{enumerate}
  \item[(i)] $0 \in I$,
  \item[(ii)] If $p,q \in I$, then $p+q \in I$,
  \item[(iii)] If $p \in I$ and $ q \in k[x_{1}, \dotsc , x_{n}]$, then $pq \in I$.
\end{enumerate}
Such a subset $I$ of $k[x_{1}, \dotsc , x_{n}]$ (or more generally any commutative ring) is called an \emph{ideal}. In the example above the ideal $I$ is (finitely) \emph{generated} by $p_{1}, \dotsc p_{r}$ (i.e. every $p \in I$ can be written as a finite sum $\sum_{j=1}^{s} f_{j}p_{j}$ where $f_{j}\in k[x_{1}, \dotsc , x_{n}]$) and we write $I = <p_{1}, \dotsc p_{r}>$.

To solve a system of linear equations in $k^{n}$, a standard algorithm is to reduce the matrix of the system to a triangular form (a normal form, even if not unique). A \emph{Groebner basis} of an ideal $I$ is a set of generators for $I$ which is a normal form for an algebraic system, in a certain sense analog to a triangular system for a set of linear equations. It permits to compute and answer effectively a lot of problems related to ideals $k[x_{1}, \dotsc , x_{n}]$ and algebraic systems. It solves, in particular, the implicitization problem. We refer to \cite{CLO07} for an excellent exposition on Groebner basis (which is very accessible for non specialists).

% ----------------------------------------------------------------

\subsection{Separation of orbits}
\label{subsec:separation_orbits}

Given an algebra $\mathcal{A}$ of invariant functions defined on $V$, we say that $\mathcal{A}$ \emph{separates} the orbits if given two vectors $v_{1}, v_{2} \in V^{2}$ which are not on the same orbit, it is always possible to find a function $f \in \mathcal{A}$ such that $f(v_{1}) \ne f(v_{2})$.
In general, when the group $G$ is not compact, the algebra of polynomial invariants does not separate the orbits\footnote{This can be illustrated by the example of the \emph{adjoint action} of the general linear group $\GL(n,\CC)$ on its Lie algebra $\gl(n,\CC)$. Each orbit is described by a unique \emph{Jordan form} but the algebra of polynomial invariants, which is generated by $\tr (M), \tr (M^{2}), \dotsc \tr (M^{n})$ where $M\in \gl(n,\CC)$, is not able to distinguish between two different Jordan forms which have the same diagonal part.}. However, it is remarkable that for any compact group, the algebra of real invariant polynomials always separates the orbits. This result is attributed to Weyl \cite{Wey97}.

\begin{prop}
Let $(V, \rho)$ be a finite dimensional, linear representation of a \emph{compact group} $G$. Then the algebra of \emph{real invariant polynomials} $\RR[V]^{G}$ separates the orbits.
\end{prop}

\begin{proof}
Let $v\in V$ and $K$ a compact subset of $V$. Given a norm $\norm{\cdot}$ on $V$, we define the distance of $v$ to $K$ as
\begin{equation*}
    d(v,K) = \min_{w\in K} \norm{v-w}.
\end{equation*}
Now, let $K_{1}$ and $K_{2}$ be two distinct $G$-orbits in $V$. The following function
\begin{equation*}
    f(v) = \frac{d(v,K_{1}) - d(v,K_{2}) }{d(v,K_{1}) + d(v,K_{2})}
\end{equation*}
is continuous on $V$, $f(v) = -1$, for $v \in K_{1}$ and $f(v) = + 1$, for $v \in K_{2}$.

Since the group $G$ is compact, the orbits are compact and it is possible to find a closed ball $B$ in $V$ which contains both orbits. We apply then the \emph{Stone-Weierstrass theorem} which states that \emph{real} polynomial functions are dense in the space of continuous functions defined on $B$ (for the $\sup$ norm). Therefore, a \emph{polynomial function $P$ on $V$ can be found such that}
\begin{equation*}
    P(v) \le - \frac{1}{2}, \: \forall v \in K_{1}, \qquad P(v) \ge \frac{1}{2}, \forall v \in K_{2} .
\end{equation*}
Now, averaging this polynomial $P$
\begin{equation*}
    \bar{P} = \int_{G} g \cdot P \, d\mu
\end{equation*}
where $\mu$ is the Haar measure over $G$, we obtain a $G$-invariant polynomial $\bar{P}$ which separates the orbits $K_{1}$ and $K_{2}$.\qed
\end{proof}

% ----------------------------------------------------------------
% ----------------------------------------------------------------

\section{Normalizers of closed subgroups of $\SO(3)$}
\label{sec:normalizers}

Every closed subgroup of $\SO(3)$ is conjugate to one of the following list \cite{DFN92,GSS88}
\begin{equation}\label{eq:closed_subgroups}
    \SO(3),\, \OO(2),\, \SO(2),\, \DD_{n} (n \ge 2),\, \ZZ_{n} (n \ge 2),\, \tetra,\, \octa,\, \ico,\, \text{and}\, \triv
\end{equation}
where:
\begin{itemize}
  \item $\OO(2)$ is the subgroup generated by all the rotations around the $z$-axis and the order 2 rotation $\sigma : (x,y,z)\mapsto (x,-y,-z)$ around the $x$-axis.
  \item $\SO(2)$ is the subgroup of all the rotations around the $z$-axis.
  \item $\ZZ_{n}$ is the unique cyclic subgroup of order $n$ of $\SO(2)$, the subgroup of rotations around the $z$-axis.
  \item $\DD_{n}$ is the \emph{dihedral} group. It is generated by $\ZZ_{n}$ and $\sigma :(x,y,z)\mapsto (x,-y,-z)$.
  \item $\tetra$ is the \emph{tetrahedral} group, the (orientation-preserving) symmetry group of a tetrahedron. It has order 12.
  \item $\octa$ is the \emph{octahedral} group, the (orientation-preserving) symmetry group of a cube or octahedron. It has order 24.
  \item $\ico$ is the \emph{icosahedral} group, the (orientation-preserving) symmetry group of a icosahedra or dodecahedron. It has order 60.
  \item $\triv$ is the trivial subgroup, containing only the unit element.
\end{itemize}

The \emph{lattice} of conjugacy classes of closed subgroups of $\SO(3)$ is completely described \cite{GSS88} by the following inclusion of subgroups (which generates order relations between conjugacy classes)
\begin{align*}
    & \ZZ_{n} \subset \DD_{n} \subset \OO(2) \qquad (n \ge 2), \\
    & \ZZ_{n} \subset \ZZ_{m} \text{ and } \DD_{n} \subset \DD_{m}, \qquad (\text{if $n$ divides $m$}), \\
    & \ZZ_{2} \subset \DD_{n} \qquad (n \ge 2), \\
    & \ZZ_{n} \subset \SO(2) \subset \OO(2) \qquad (n \ge 2),
\end{align*}
and the arrows in figure~\ref{fig:lattice2} which complete the lattice, taking account of the exceptional subgroups $\octa, \tetra, \ico$ (beware that figure~\ref{fig:lattice2} cannot be realized by an inclusion diagram between representatives of these conjugacy classes).

\begin{figure}[h]
  \centering
  \includegraphics{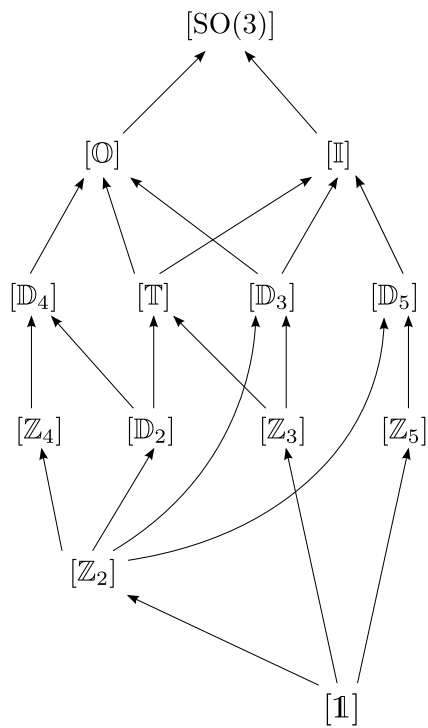}
  \caption{Exceptional classes in the lattice of conjugacy classes of closed subgroup of $\SO(3)$.}
  \label{fig:lattice2}
\end{figure}

\begin{rem}
Notice that $\octa$ and $\ico$ are maximal subgroups of $\SO(3)$ which both contained $\tetra$.
\end{rem}

Recall that the \emph{normalizer} of a subgroup $H$ of a group $G$, denoted by $N(H)$ is the biggest subgroup of $G$ in which $H$ is normal. It is defined by
\begin{equation*}
    N(H) := \set{g\in G;\; gHg^{-1} = H}.
\end{equation*}

\begin{prop}
The normalizers of the subgroups in the list~\eqref{eq:closed_subgroups} are given below.
\begin{align*}
    N(\SO(3)) &= \SO(3), & N(\OO(2)) &= \OO(2), & N(\SO(2)) &= \OO(2), \\
    N(\octa) &= \octa, & N(\ico) &= \ico, & N(\tetra) &= \octa, \\
    N(\triv) &= \SO(3), & N(\ZZ_{n}) &= \OO(2) \text{ for } n\ge 2, & N(\DD_{n}) &= \DD_{2n} \text{ for } n \ge 3,
\end{align*}
and $N(\DD_{2}) = \octa$. Moreover, the quotient groups are given by
\begin{align*}
    N(\SO(3))/\SO(3) & = \triv, & N(\OO(2))/\OO(2) & = \triv, & N(\SO(2))/\SO(2) & = \ZZ_{2}, \\
    N(\octa)/\octa & = \triv, & N(\ico)/\ico &= \triv, & N(\tetra)/\tetra &= \ZZ_{2}, \\
    N(\triv)/\triv &= \SO(3), & N(\ZZ_{n})/\ZZ_{n} & = \OO(2) \text{ for } n\ge 2, & N(\DD_{n})/\DD_{n} & = \ZZ_{2} \text{ for } n \ge 3,
\end{align*}
and $N(\DD_{2})/\DD_{2} = \GS_{3}$, the symmetric group of 3 elements.
\end{prop}

\begin{proof}
That $N(\SO(3)) = \SO(3)$ and $N(\triv) = \SO(3)$ are obvious. The normalizers of $\OO(2)$, $\octa$ and $\ico$ are equal to themselves because each of these groups is a \emph{maximal} subgroup of $\SO(3)$ and because $\SO(3)$ is a \emph{simple} group (it has no non-trivial \emph{normal} subgroups). Given any rotation $r_{\theta} \in \SO(2)$, we have
\begin{equation*}
    \sigma r_{\theta} \sigma = r_{-\theta}
\end{equation*}
where $\sigma :(x,y,z)\mapsto (x,-y,-z)$. Hence $\OO(2) \subset N(\SO(2))$ but since $\OO(2)$ is maximal and $\SO(3)$ is simple we get that $N(\SO(2)) = \OO(2)$. $\tetra$ is a subgroup of index 2 of $\octa$ and is thus a normal subgroup of $\octa$. Therefore, we have $\octa \subset N(\tetra)$ but since $\octa$ is maximal and $\SO(3)$ is simple, this proves that $N(\tetra) = \octa$. Next consider $\ZZ_{n}$ for some $n\ge 2$. Since $\ZZ_{n}$ is a subgroup of the abelian group $\SO(2)$ we have $\SO(2) \subset N(\ZZ_{n})$. Moreover
\begin{equation*}
    \sigma r_{2\pi/n} \sigma = r_{-2\pi/n} \in \ZZ_{n}
\end{equation*}
and therefore $\OO(2) \subset N(\ZZ_{n})$. Using again the maximality argument, we get $N(\ZZ_{n}) = \OO(2)$. Finally, consider $\DD_{n}$ for some $n\ge 2$. Because $\DD_{n}$ is subgroup of index 2 of $\DD_{2n}$, we have $\DD_{2n} \subset N(\DD_{n})$. Suppose first that $n > 2$ and let $g \in N(\DD_{n})$. Then $gr_{2\pi/n}g^{-1}$ is a rotation of angle $\pm 2\pi/n$ and its axis is therefore the $z$-axis. Hence $g(e_{3}) = \pm e_{3}$. Besides, $g\sigma g^{-1}$ is a rotation of order 2 whose axis is not the $z$-axis (otherwise we would have $g(e_{1}) = \pm e_{3}$). Therefore we have
\begin{equation*}
    g\sigma g^{-1} = r_{2k\pi/n}\sigma r_{-2k\pi/n} = r_{4k\pi/n}\sigma,
\end{equation*}
for some $k \in\set{0,1, \dotsc , n-1}$. Now if $g(e_{3}) = e_{3}$, we have $g = r_{\theta}$ and necessarily $2\theta = 4k\pi/n$. If $g(e_{3}) = -e_{3}$ we have $g = \sigma r_{\theta}$ and necessarily $2\theta = 4k\pi/n$. In both cases, this means that $g \in \DD_{2n}$ and we conclude that $N(\DD_{n}) = \DD_{2n}$. Suppose now that $n=2$, in which case $\DD_{2}$ is the \emph{Klein group}
\begin{equation*}
    \DD_{2} := \set{e, R_{e_{1}}, R_{e_{2}}, R_{e_{3}}},
\end{equation*}
where $R_{e_{j}}$ denotes the rotation by angle $\pi$ around $e_{j}$. For each element $g$ of the octahedral group $\octa$, we have $g(e_{i}) = \pm e_{j}$ and hence
\begin{equation*}
    gR_{e_{i}}g^{-1} = R_{e_{j}}.
\end{equation*}
Therefore $\octa \subset N(\DD_{2})$. Using again the maximality argument, we conclude that $N(\DD_{2}) = \octa$.

For the computation of the quotient groups, all of them are obvious but $N(\ZZ_{n})/\ZZ_{n}$ and $N(\DD_{2})/\DD_{2}$. Concerning the first quotient group, we remark that every element in $\OO(2)$ can be written uniquely as
\begin{equation*}
    \sigma^{\varepsilon}r
\end{equation*}
where $\varepsilon = 0,1$ and $r\in \SO(2)$ is a rotation around the $z$-axis. Consider now the map
\begin{equation*}
    \phi : \OO(2) \to \OO(2), \qquad \sigma^{\varepsilon}r \mapsto \sigma^{\varepsilon}r^{n}.
\end{equation*}
One can check that this map is a surjective, group morphism and that its kernel is exactly $\ZZ_{n}$. This establish the isomorphism between $\OO(2)/\ZZ_{n}$ and $\OO(2)$.

For the second quotient group, we recall that there is a well-known isomorphism from the symmetry group of the cube $\octa$ onto the group of permutations of the diagonals of the cube, that is $\GS_{4}$. Under this isomorphism, the subgroup $\DD_{2}$ of $\octa$ is sent to
\begin{equation*}
    \mathcal{D}_{2} := \set{e, (12)(34), (13)(24), (23)(14)},
\end{equation*}
which is a normal subgroup of $\GS_{4}$. Let $\GS_{3}$ be the subgroup of $\GS_{4}$ which fixes the element $4$. Then $\GS_{4}$ is the semi-direct product of $\mathcal{D}_{2}$ by $\GS_{3}$. In particular, the quotient $\GS_{4}/\mathcal{D}_{2}$ is isomorphic to $\GS_{3}$. This subgroup $\GS_{3}$ corresponds geometrically to the subgroup of $\octa$ generated by a rotation of angle $2\pi/3$ around a diagonal and a rotation of angle $\pi$ around any diagonal perpendicular to the first one, that is a conjugate of the group $\DD_{3}$.
\end{proof}

% ----------------------------------------------------------------
% ----------------------------------------------------------------

\section{Trace formulas for $V^{H}$}
\label{sec:TrFrm}

Let $(V, \rho)$ be a finite dimensional linear representation of $\SO(3)$, which admits the following harmonic isotypic decomposition :
\begin{equation}\label{eq:DecHar3D}
    V\cong\bigoplus_{k=0}^{n}\alpha_{k}\HH^{k}
\end{equation}
where $\alpha_{k}$ is the multiplicity of the space $\HH^{k}$ in the decomposition.
Using the trace formula \eqref{eq:trace_formula} the following explicit formulas can easily be obtained \cite{GSS88,Auf10}
\begin{equation}
\begin{split}
  \dim V^{\ZZ_{p}} & = 2\sum_{k=0}^{n} \alpha_{k}\left[ \frac{k}{p}\right]+\sum_{k=0}^{n} \alpha_{k}\\
  \dim V^{\DD_{p}} & = \sum_{k=0}^{n} \alpha_{k}\left[ \frac{k}{p}\right]+\sum_{k=0}^{\left[ \frac{n}{2}\right]}\alpha_{2k}\\
  \dim V^{\tetra} & = \sum_{k=0}^{n} \alpha_{k}\left(2\left[ \frac{k}{3}\right]+\left[ \frac{k}{2}\right]-k+1\right)\\
  \dim V^{\octa} & = \sum_{k=0}^{n} \alpha_{k}\left(\left[ \frac{k}{4}\right]+\left[ \frac{k}{3}\right]+\left[ \frac{k}{2}\right]-k+1\right)\\
  \dim V^{\ico} & = \sum_{k=0}^{n} \alpha_{k}\left(\left[ \frac{k}{5}\right]+\left[ \frac{k}{3}\right]+\left[ \frac{k}{2}\right]-k+1\right)\\
\end{split}
\end{equation}

% ----------------------------------------------------------------
% ----------------------------------------------------------------

\end{document}